\documentclass[journal,onecolumn,10pt,twoside]{IEEEtran}

\usepackage{etex}
\usepackage{amsmath}
\usepackage{graphicx}
\usepackage{framed, color} 
\usepackage{pdflscape}
\definecolor{shadecolor}{rgb}{1, 0.8, 0.3}

\usepackage{amssymb}
\usepackage{amsthm}
\makeatletter
\def\th@plain{%
  \thm@notefont{}% same as heading font
  \itshape % body font
}
\def\th@definition{%
  \thm@notefont{}% same as heading font
  \normalfont % body font
}
\makeatother

\usepackage{amsfonts}
\usepackage{cite}
\usepackage{ifthen}
\usepackage{epsfig}
\usepackage{caption}
\usepackage{subcaption}
\usepackage{array}
\usepackage{enumerate}
\usepackage{algorithm}
\usepackage[all]{xy}
\usepackage{algpseudocode}

\newcommand{\subparagraph}{}
\usepackage[compact]{titlesec}

\usepackage{times}
\usepackage{float}

\usepackage{nth}
\usepackage{array} 
\usepackage{url,soul}
\usepackage{wrapfig}

\theoremstyle{plain}
\newtheorem{theorem}{Theorem}
\newtheorem{lemma}[theorem]{Lemma}
\newtheorem{proposition}[theorem]{Proposition}
\newtheorem{corollary}[theorem]{Corollary}

\theoremstyle{definition}
\newtheorem{definition}{Definition}
\newtheorem{example}{Example}

\theoremstyle{remark}
\newtheorem{remark}{Remark}

\newcommand{\bal}{\begin{align}}
\newcommand{\eal}{\end{align}}	

\newcommand{\beq}{\begin{equation}}
\newcommand{\eeq}{\end{equation}}
\newcommand{\bea}{\begin{eqnarray}}
\newcommand{\eea}{\end{eqnarray}}
\newcommand{\bean}{\begin{eqnarray*}}
\newcommand{\eean}{\end{eqnarray*}}
\newcommand{\bit}{\begin{itemize}}
\newcommand{\eit}{\end{itemize}}
\newcommand{\ben}{\begin{enumerate}}
\newcommand{\een}{\end{enumerate}}
\newcommand{\blem}{\begin{lem}}
\newcommand{\elem}{\end{lem}}
\newcommand{\bthm}{\begin{thm}}
\newcommand{\ethm}{\end{thm}}
\newcommand{\bpf}{\begin{IEEEproof}}
\newcommand{\epf}{\end{IEEEproof}}

\newcommand\defeq{\mathrel{\overset{\makebox[0pt]{\mbox{\normalfont\tiny\sffamily def}}}{=}}}

\newcommand*\thth[1]{$#1^{\text{\tiny th}}$}
\newcommand{\EE}{\ensuremath{\mathbb{E}}}
\newcommand{\PP}
{\Pr}

\newcommand{\ba}{\mathbf A}
\newcommand{\bsa}{\mathbf a}
\newcommand{\bx}{\mathbf X}
\newcommand{\bsx}{\mathbf x}
\newcommand{\bsy}{\mathbf y}
\newcommand{\bsz}{\mathbf z}
\newcommand{\bsr}{\mathbf r}

\ifdefined\showredblue
	
\fi

\ifdefined\showblue
	
\fi

\ifdefined\shownothing
	
\fi

\title{Speeding Up Distributed Machine Learning\\Using Codes}
\author{
Kangwook Lee, Maximilian Lam, Ramtin Pedarsani, \\
Dimitris Papailiopoulos, and Kannan Ramchandran, \emph{Fellow, IEEE}
\thanks{Kangwook Lee is with the School of Electrical Engineering, KAIST. Maximilian Lam and Kannan Ramchandran are with the Department of Electrical Engineering and Computer Sciences, University of California, Berkeley. Ramtin Pedarsani is with the Department of Electrical and Computer Engineering, University of California, Santa Barbara. Dimitris Papailiopoulos is with the Department of Electrical and Computer Engineering, University of Wisconsin-Madison. 
Emails: kw1jjang@kaist.ac.kr, agnusmaximus@berkeley.edu, ramtin@ece.ucsb.edu, dimitris@papail.io, kannanr@eecs.berkeley.edu.

This work is published in IEEE Transactions on Information Theory~\cite{itjournal}, and was presented in part at the 2015 Neural Information Processing Systems (NIPS) Workshop on Machine Learning Systems~\cite{lee2015nips}, and the 2016 IEEE International Symposium on Information Theory (ISIT)~\cite{lee2016isit}.}}

\begin{document}
\maketitle
\begin{abstract}
Codes are widely used in many engineering applications to offer {\it robustness} against {\it noise}.
In large-scale systems there are several types of noise that can affect the performance of distributed machine learning algorithms -- straggler nodes, system failures, or communication bottlenecks -- but there has been little interaction cutting across codes, machine learning, and distributed systems.
In this work, we provide theoretical insights on how {\it coded} solutions can achieve significant gains compared to uncoded ones.
We focus on two of the most basic building blocks of distributed learning algorithms: {\it matrix multiplication} and {\it data shuffling}.
For matrix multiplication, we use codes to alleviate the effect of stragglers, and show that if the number of homogeneous workers is $n$, and the runtime of each subtask has an exponential tail, coded computation can speed up distributed matrix multiplication by a factor of $\log n$.
For data shuffling, we use codes to reduce communication bottlenecks, exploiting the excess in storage. 
We show that when a constant fraction $\alpha$ of the data matrix can be cached at each worker, and $n$ is the number of workers, \emph{coded shuffling} reduces the communication cost by a factor of $(\alpha + \frac{1}{n})\gamma(n)$ compared to uncoded shuffling, where $\gamma(n)$ is the ratio of the cost of unicasting $n$ messages to $n$ users to multicasting a common message (of the same size) to $n$ users.
For instance, $\gamma(n) \simeq n$ 
if multicasting a message to $n$ users is as cheap as unicasting a message to one user.
We also provide experiment results, corroborating our theoretical gains of the coded algorithms.
\end{abstract}
\section{Introduction}\label{sec:intro}

In recent years, the computational paradigm for large-scale machine learning and data analytics has shifted towards massively large distributed systems, comprising individually small and unreliable computational nodes (low-end, commodity hardware).  
Specifically, modern distributed systems like Apache Spark \cite{apache_spark} and computational primitives like MapReduce \cite{mapreduce} have gained significant traction, 
as they enable the execution of production-scale tasks on data sizes of the order of petabytes.
However, it is observed that the performance of a modern distributed system is significantly affected by anomalous system behavior and bottlenecks \cite{40801}, i.e., a form of ``system noise". 
Given the individually unpredictable nature of the nodes in these systems, we are faced with the challenge of securing fast and high-quality algorithmic results in the face of uncertainty.

\begin{figure}[t]
\centering
\includegraphics[width= 0.5\textwidth]{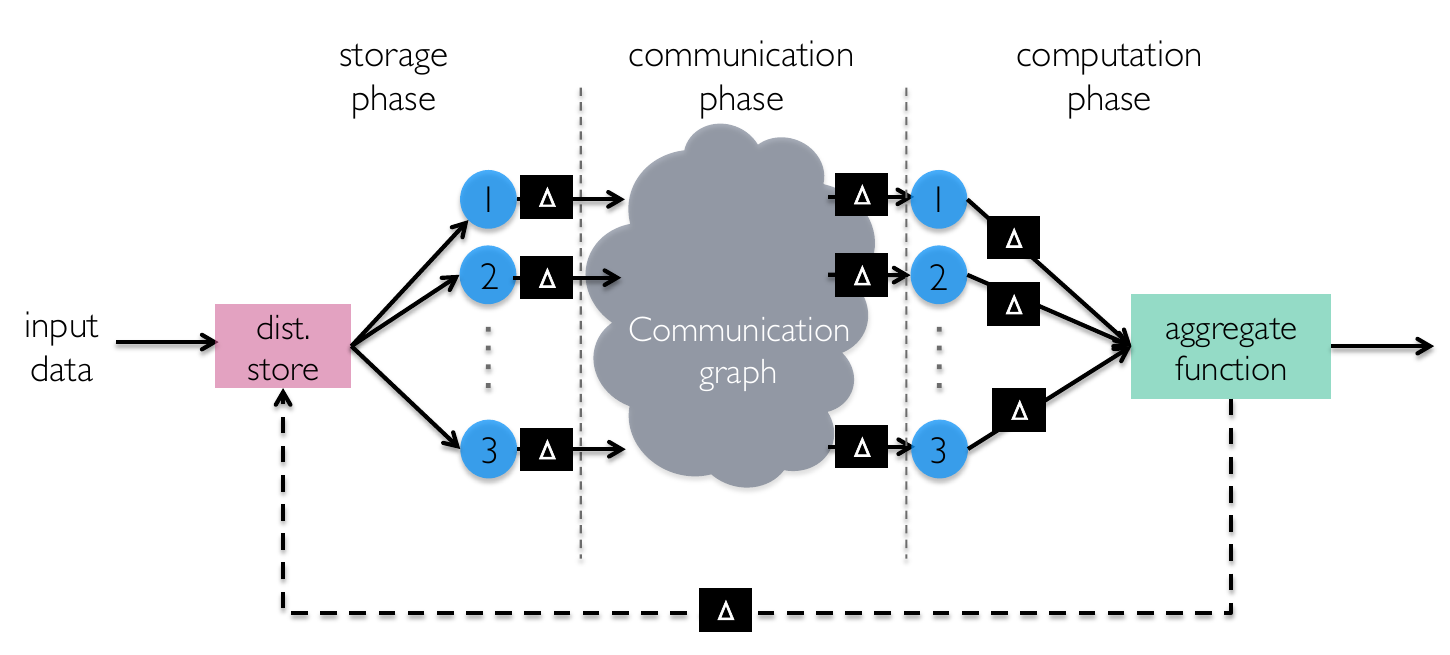}
\caption{\footnotesize{\textbf{Conceptual diagram of the phases of distributed computation.}
The algorithmic workflow of distributed (potentially iterative) tasks, can be seen as receiving input data, storing them in distributed nodes, communicating data around the distributed network, and then computing locally a function at each distributed node. The main bottlenecks in this execution (communication, stragglers, system failures) can all be abstracted away by incorporating a notion of delays between these phases, denoted by $\Delta$ boxes.
\label{fig:block_diagram}}}
\end{figure}
In this work, we tackle this challenge using \emph{coding theoretic} techniques. 
The role of codes in providing resiliency against noise has been studied for decades in several other engineering contexts, and is part of our everyday infrastructure (smartphones, laptops, WiFi and cellular systems, etc.).
The goal of our work is to apply coding techniques to blueprint robust distributed systems, especially for distributed machine learning algorithms. 
The workflow of distributed machine learning algorithms in a large-scale system can be decomposed into three functional phases: a storage, a communication, and a computation phase, as shown in Fig.~\ref{fig:block_diagram}.
In order to develop and deploy sophisticated solutions and tackle large-scale problems in machine learning, science, engineering, and commerce, it is important to understand and optimize novel and complex trade-offs across the multiple dimensions of computation, communication, storage, and the accuracy of results.   
Recently, codes have begun to transform the storage layer of distributed systems in modern data centers under the umbrella of regenerating and locally repairable codes for distributed storage
\cite{dimakis2010network, 
rashmi2011optimal, 
suh2011exact,  
tamo2011mds, 
cadambe2011optimal,  
papailiopoulos2011repair, 
gopalan2011locality, 
oggier2011self,
papailiopoulos2012simple,
han2007reliable,
huang2007pyramid,
papailiopoulos2012locally,
kamath2012codes,
rawat2012optimal,
prakash2012optimal,
silberstein2012error}
 which are also having a major impact on industry~\cite{huang2012erasure,sathiamoorthy2013xoring,rashmi2013hotstorage, rashmi2014hitchhiker}.

In this paper, we explore the use of coding theory to remove bottlenecks caused during the other phases: \emph{the communication and computation phases} of distributed algorithms. 
More specifically, we identify two core blocks relevant to the communication and computation phases that we believe are key primitives in a plethora of distributed data processing and machine learning algorithms: \emph{matrix multiplication} and \emph{data shuffling}.

%
%
%
%
%\begin{figure}[t]
%\centering
%\includegraphics[width= 0.5\textwidth]{figs/stragglers.pdf}
%\caption{\footnotesize{\textbf{The effects of slow nodes.} In distributed computation, the running time of a single distributed task is governed by that of the slowest node.
%In this toy figure, we see how slow nodes can significantly  impact the running time of distributed computation. Can we use coding to alleviate the straggler's effects?\label{fig:stragglers}}}
%\end{figure}
For matrix multiplication, we use codes to leverage the plethora of nodes and alleviate the effect of \emph{stragglers}, i.e., nodes that are significantly slower than average.
We show analytically that if there are $n$ workers having identically distributed computing time statistics that are exponentially distributed, the optimal \emph{coded matrix multiplication} is $\Theta(\log n)$\footnote{For any two sequences $f(n)$ and $g(n)$: $f(n) = \Omega(g(n))$ if there exists a positive constant $c$ such that $f(n)\geq cg(n)$; 
$f(n) = o(g(n))$ if $\lim_{n\rightarrow \infty} \frac{f(n)}{g(n)} =0$.} times faster than the uncoded matrix multiplication on average.

Data shuffling is a core element of many machine learning applications, and is well-known to improve the statistical performance of learning algorithms.
We show that codes can be used in a novel way to trade off excess in available storage for reduced communication cost for data shuffling done in parallel machine learning algorithms.
We show that when a constant fraction of the data matrix can be cached at each worker, and $n$ is the number of workers, \emph{coded shuffling} reduces the communication cost by a factor $\Theta(\gamma(n))$ compared to uncoded shuffling, where $\gamma(n)$ is the ratio of the cost of unicasting $n$ messages to $n$ users to multicasting a common message (of the same size) to $n$ users.
For instance, $\gamma(n) \simeq n$ if multicasting a message to $n$ users is as cheap as unicasting a message to one user.

We would like to remark that a major innovation of our coding solutions is that they are woven into the fabric of the algorithmic design, and coding/decoding is performed over the representation field of the input data (e.g., floats or doubles).
In sharp contrast to most coding applications, we do not need to ``re-factor code" and modify the distributed system to accommodate for our solutions; it is all done seamlessly in the algorithmic design layer, an abstraction that we believe is much more impactful as it is located ``higher up'' in the system layer hierarchy compared to traditional applications of coding that need to interact with the stored and transmitted ``bits" (e.g., as is the case for coding solutions for the physical or storage layer).
%Further, we provide experiment results, which show potential impacts of our coded solutions on the design of modern distributed systems in practice. 

\subsection{Overview of The Main Results} \label{sec:overview}
We now provide a brief overview of the main results of this paper. 
\begin{figure}[h]
\centering
\includegraphics[width=.48\textwidth]{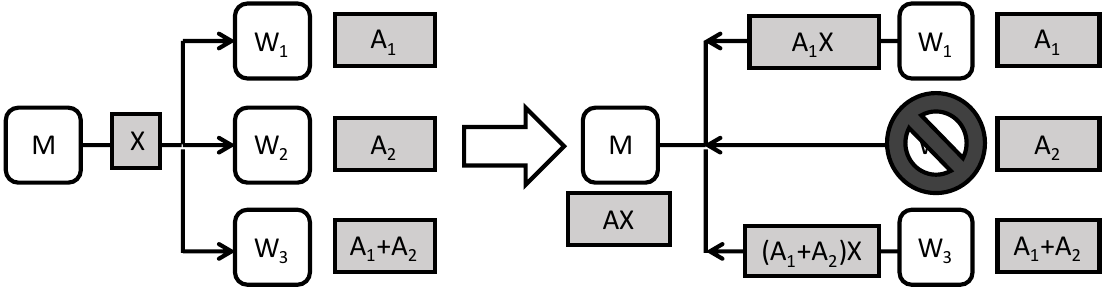}
\caption{\footnotesize{\textbf{Illustration of \emph{Coded Matrix Multiplication}.} Data matrix $\ba$ is partitioned into $2$ submatrices: $\ba_1$ and $\ba_2$. Node $W_1$ stores $\ba_1$, node $W_2$ stores $\ba_2$, and node $W_3$ stores $\ba_1 + \ba_2$. Upon receiving $\bx$, each node multiplies $\bx$ with the stored matrix, and sends the product to the master node. Observe that the master node can always recover $\ba \bx$ upon receiving \emph{any} $2$ products, without needing to wait for the slowest response. For instance, consider a case where the master node has received $\ba_1\bx$ and $(\ba_1 + \ba_2)\bx$. By subtracting $\ba_1 \bx$ from $(\ba_1 + \ba_2)\bx$, it can recover $\ba_2 \bx$ and hence $\ba\bx$.}}
\label{fig:mm_system}
\end{figure}
The following toy example illustrates the main idea of \emph{Coded Computation}.
Consider a system with three worker nodes and one master node, as depicted in Fig.~\ref{fig:mm_system}. 
The goal is to compute a matrix multiplication $\ba\bx$ for data matrix $\ba\in \mathbb{R}^{q \times r}$ and input matrix $\bx\in \mathbb{R}^{r \times s}$.
The data matrix $\ba$ is divided into two submatrices $\ba_1\in \mathbb{R}^{q/2 \times r}$ and $\ba_2\in \mathbb{R}^{q/2 \times r}$ and stored in node $1$ and node $2$, as shown in Fig.~\ref{fig:mm_system}.
The sum of the two submatrices is stored in node $3$.
After the master node transmits $\bx$ to the worker nodes, each node computes the matrix multiplication of the stored matrix and the received matrix $\bx$, and sends the computation result back to the master node.
The master node can compute $\ba\bx$ as soon as it receives \emph{any} two computation results.

\emph{Coded  Computation} designs parallel tasks for a linear operation using erasure codes such that its runtime is not affected by up to a certain number of stragglers. 
Matrix multiplication is one of the most basic linear operations and is the workhorse of a host of machine learning and data analytics algorithms, e.g., gradient descent based algorithm for regression problems, power-iteration like algorithms for spectral analysis and graph ranking applications, etc. 
Hence, we focus on the example of matrix multiplication in this paper.
With coded computation, we will show that the runtime of the algorithm can be significantly reduced compared to that of other uncoded algorithms. 
The main result on \emph{Coded Computation} is stated in the following (informal) theorem.

\begin{theorem}[Coded computation]
If the number of workers is $n$, and the runtime of each subtask has an exponential tail, the optimal coded matrix multiplication is $\Theta(\log n)$ times faster than the uncoded matrix multiplication.
\end{theorem}
For the formal version of the theorem and its proof, see Sec.~\ref{sec:shifted_exp}.

%\subsection{Coded Shuffling}
We now overview the main results on coded shuffling.
Consider a master-worker setup where a master node holds the entire data set.
The generic machine learning task that we wish to optimize is the following: 
1) the data set is randomly permuted and partitioned in batches at the master;
2) the master sends the batches to the workers;
3) each worker uses its batch and locally trains a model;
4) the local models are averaged at the master and the process is repeated. 
To reduce communication overheads between master and workers, {\it Coded Shuffling} exploits {\it i)} the locally cached data points of previous passes and {\it ii)} the ``transmission strategy" of the master node.

\begin{figure}[h]
\centering
\includegraphics[width=.48\textwidth]{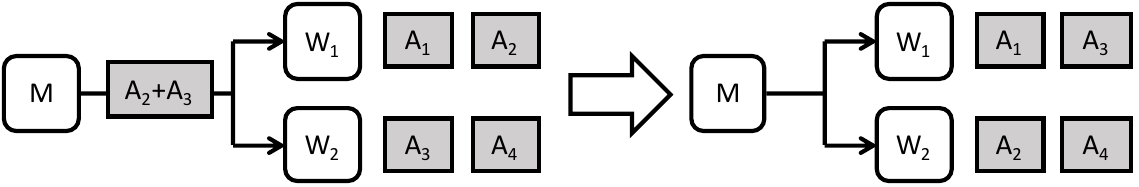}
\caption{\footnotesize{\textbf{Illustration of \emph{Coded Shuffling}.} Data matrix $\ba$ is partitioned into $4$ submatrices: $\ba_1$ to $\ba_4$. Before shuffling, worker $W_1$ has $\ba_1$ and $\ba_2$ and worker $W_2$ has $\ba_3$ and $\ba_4$. The master node can send $\ba_2 + \ba_3$ in order to shuffle the data stored at the two workers.} }
\label{fig:coded_shuffling_illustration}
\end{figure}

We illustrate the basics of \emph{Coded Shuffling} with a toy example.
Consider a system with two worker nodes and one master node.
Assume that the data set consists of 4 batches ${\bf A}_1,\ldots,{\bf A}_4$, which are stored across two workers as shown in Fig.~\ref{fig:coded_shuffling_illustration}.
The sole objective of the master is to transmit ${\bf A}_3$ to the first worker and ${\bf A}_4$ to the second.
For this purpose, the master node can simply multicast a \emph{coded} message $\ba_2 + \ba_3$ to the worker nodes since the workers can decode the desired batches using the stored batches.
Compared to the na\"{\i}ve (or uncoded) shuffling scheme in which the master node transmits $\ba_2$ and $\ba_3$ separately, this new shuffling scheme can save $50\%$ of the communication cost, speeding up the overall machine learning algorithm.
The {\it Coded Shuffling} algorithm is a generalization of the above toy example, which we explain in detail in Sec.~\ref{sec:shuffling}.

Note that the above example assumes that {\it multicasting} a message to all workers costs exactly the same as unicasting a message to one of the workers. 
In general, we capture the advantage of using multicasting over unicasting by defining  $\gamma(n)$ as follows:
\begin{align}
&\gamma(n) \nonumber\\
&\defeq \frac{\text{cost of unicasting}~n~\text{separate msgs to}~n~\text{workers}}{\text{cost of multicasting a common msg to}~n~\text{workers}}.
\end{align}
Clearly, $1 \leq \gamma(n) \leq n$: if $\gamma(n)=n$, the cost of multicasting is equal to that of unicasting a single message (as in the above example); if $\gamma(n)=1$, there is essentially no advantage of using multicast over unicast.

We now state the main result on \emph{Coded Shuffling} in the following (informal) theorem.
\begin{theorem}[Coded shuffling]
Let $\alpha$ be the fraction of the data matrix that can be cached at each worker, and $n$ be the number of workers. 
Assume that the advantage of multicasting over unicasting is $\gamma(n)$. 
Then, coded shuffling reduces the communication cost by a factor of $\left(\alpha + \frac{1}{n}\right) \gamma(n)$ compared to uncoded shuffling.
\end{theorem}
For the formal version of the theorem and its proofs, see Sec.~\ref{sec:shuffle_main_results}.

The remainder of this paper is organized as follows. In Sec.~\ref{sec:related}, we provide an extensive review of the related works in the literature. Sec.~\ref{sec:computation} introduces the coded matrix multiplication, and Sec.~\ref{sec:shuffling} introduces the coded shuffling algorithm.
Finally, Sec.~\ref{sec:conclusion} presents conclusions and discusses open problems.

\section{Related Work}\label{sec:related}
\subsection{Coded Computation and Straggler Mitigation}
The straggler problem has been widely observed in distributed computing clusters. 
The authors of \cite{40801} show that running a computational task at a computing node often involves unpredictable latency due to several factors such as network latency, shared resources, maintenance activities, and power limits. Further, they argue that stragglers cannot be completely removed from a distributed computing cluster.
The authors of \cite{Mantri} characterize the impact and causes of stragglers that arise due to resource contention, disk failures, varying network conditions, and imbalanced workload. 

One approach to mitigate the adverse effect of stragglers is based on efficient straggler detection algorithms. 
For instance, the default scheduler of Hadoop constantly detects stragglers while running computational tasks. Whenever it detects a straggler, it relaunches the task that was running on the detected straggler at some other available node.  
In \cite{zaharia_late}, Zaharia et al. propose a modification to the existing straggler detection algorithm and show that the proposed solution can effectively reduce the completion time of MapReduce tasks.
In \cite{Mantri}, Ananthanarayanan et al. propose a system that efficiently detects stragglers using real-time progress and cancels those stragglers, and show that the proposed system can further reduce the runtime of MapReduce tasks.

Another line of work is based on breaking the synchronization barriers in distributed algorithms\,\cite{agarwal2011distributed, HogWild!}.
An asynchronous parallel execution can continuously make progress without having to wait for all the responses from the workers, and hence the overall runtime is less affected by stragglers.
However, these asynchronous approaches break the serial consistency of the algorithm to be parallelized, and do not guarantee ``correctness'' of the end result, i.e., the output of the asynchronous algorithm can differ from that of a serial execution with an identical number of iterations.

Recently, replication-based approaches have been explored to tackle the straggler problem: by replicating tasks and scheduling the replicas, the runtime of distributed algorithms can be significantly improved \cite{ananthanarayanan2013effective, 6736597, wang2014efficient, mor_exact_analysis, chaubeyreplicated, rr_w_cancellation_cost, joshitompecs}. 
By collecting outputs of the fast-responding nodes (and potentially canceling all the other slow-responding replicas), such replication-based scheduling algorithms can reduce latency. 
In \cite{rr_w_cancellation_cost}, the authors show that even without replica cancellation, one can still reduce the average task latency by properly scheduling redundant requests. 
We view these policies as special instances of coded computation: such task replication schemes can be seen as \emph{repetition-coded} computation. In Sec.~\ref{sec:computation}, we describe this connection in detail, and indicate that coded computation can significantly outperform replication (as is usually the case for coding vs. replication in other engineering applications).

Another line of work that is closely related to coded computation is about the latency analysis of coded distributed storage systems. 
In \cite{huang2012codes, mdsqueuejournal}, the authors show that the flexibility of erasure-coded distributed storage systems allows for faster data retrieval performance than replication-based distributed storage systems. 
Joshi et al.\,\cite{joshi2014} show that scheduling redundant requests to an increased number of storage nodes can improve the latency performance, and characterize the resulting storage-latency tradeoff. 
Sun et al.\,\cite{sun2015provably} study the problem of adaptive redundant requests scheduling, and characterize the optimal strategies for various scenarios.
In \cite{kadheallerton, kadheisit}, Kadhe and Soljanin analyze the latency performance of \emph{availability codes}, a class of storage codes designed for enhanced availability. 
In \cite{joshitompecs}, the authors study the cost associated with scheduling of redundant requests, and propose a general scheduling policy that achieves a delicate balance between the latency performance and the cost.

We now review some recent works on coded computation, which have been published after our conference publications\,\cite{lee2015nips, lee2016isit}.
In~\cite{nuwan2016anytime}, an anytime coding scheme for approximate matrix multiplication is proposed, and it is shown that the proposed scheme can improve the quality of approximation compared with the other existing coded schemes for exact computation. 
In \cite{dutta2016short}, the authors propose a coded computation scheme called `Short-Dot'.
Short-Dot induces additional sparsity to the encoded matrices at the cost of reduced decoding flexibility, and hence potentially speeds up the computation. 
The authors of\,\cite{tandon2016gradient} consider the problem of computing gradients in a distributed system, and propose a novel coded computation scheme tailored for computing a sum of functions.
In many machine learning problems, the objective function is a sum of per-data loss functions, and hence the gradient of the objective function is the sum of gradients of per-data loss functions.
Based on this observation, they propose \emph{Gradient Coding}, which can reliably compute the exact gradient of any function in the presence of stragglers.
While Gradient coding can be applied to computing gradients of any functions, it usually incurs significant storage and computation overheads. 
In\,\cite{bitar2017minimizing}, the authors consider a secure coded computation problem where the input data matrices need to be secured from the workers.
They propose a secure computation scheme based on Staircase codes, which can speed up the distributed computation while securing the input data from the workers.  
In \cite{lee2017isitblockcode}, the authors consider the problem of large matrix-matrix multiplication, and propose a new coded computation scheme based on product codes. 
In \cite{codedcomputationhet}, the authors consider the coded computation problem on heterogenous computing clusters while our work assumes a homogeneous computing cluster.
The authors show that by delicately distributing jobs across heterogenous workers, one can improve the performance of coded computation compared with the symmetric job allocation scheme, which is designed for homogeneous workers in our work. 
While most of the works focus on the application of coded computation to \emph{linear} operations, a recent work shows that coding can be used also in distributed computing frameworks involving \emph{nonlinear} operations\,\cite{lee2017isitmulticore}.
The authors of \cite{lee2017isitmulticore} show that by leveraging the multi-core architecture in the worker computing units and ``coding across'' the multi-core computed outputs, significant (and in some settings unbounded) gains in speed-up in computational time can be achieved between the coded and uncoded schemes.

\subsection{Data Shuffling and Communication Overheads}
Distributed learning algorithms  on large-scale networked systems have been extensively studied in the literature 
\cite{bertsekasbook, nedic2009distributed, boyd2011distributed, bekkerman2011scaling, duchi2012dual, chen2012diffusion, dean2012large, low2012distributed, kraska2013mlbase, sparks2013mli, li2014scaling}.
Many of the distributed algorithms that are implemented in practice share a similar algorithmic ``anatomy": 
the data set is split among several cores or nodes, each node trains a model locally, then the local models are averaged, and the process is repeated. 
While training a model with parallel or distributed learning algorithms, it is common to randomly re-shuffle the data a number of times \cite{JellyFish, HogWild!, bottou2012stochastic, DimmWitted,gurbuzbalaban2015random, ioffe2015batch}.
This essentially means that after each shuffling the learning algorithm will go over the data in a different order than before.
Although the effects of random shuffling are far from understood theoretically, the large statistical gains have turned it into a common practice.
Intuitively, data shuffling before a new pass over the data, implies that nodes get a nearly ``fresh" sample from the data set, which experimentally leads to better statistical performance.
Moreover, bad orderings of the data---known to lead to slow convergence in the worst case \cite{JellyFish,gurbuzbalaban2015random,ioffe2015batch}---are ``averaged out".
However, the statistical benefits of data shuffling do not come for free: each time a new shuffle is performed, the {\it entire} dataset is communicated over the network of nodes.
This inevitably leads to performance bottlenecks due to heavy communication.

In this work, we propose to use coding opportunities to significantly reduce the communication cost of some distributed learning algorithms that require data shuffling.  
Our coded shuffling algorithm is built upon the coded caching scheme by Maddah-Ali and Niesen \cite{MN1}. Coded caching is a technique to reduce the communication rate in content delivery networks. Mainly motivated by video sharing applications, coded caching exploits the multicasting opportunities between users that request different video files to significantly reduce
the communication burden of the server node that has access to the files. Coded caching has been studied in many
scenarios such as decentralized coded caching \cite{MN2}, online coded caching \cite{Online}, hierarchical coded caching for wireless communication \cite{Hierarchical}, and
device-to-device coded caching \cite{D2D}. Recently, the authors in~\cite{cmr} proposed coded MapReduce that reduces the communication cost in the process of transferring the results of mappers to reducers.

Our proposed approach is significantly different from all related studies on coded caching in two ways:
(i) we shuffle the \emph{data points} among the computing nodes to \emph{increase the statistical efficiency} of distributed computation and machine learning algorithms; 
and 
(ii) we \emph{code the data over their actual representation} (i.e., over the doubles or floats) unlike the traditional coding schemes over bits. In Sec.~\ref{sec:shuffling}, we describe how coded shuffling can remarkably speed up the communication phase of large-scale parallel machine learning algorithms, and provide extensive numerical experiments to validate our results.

The coded shuffling problem that we study is related to the index coding problem\,\cite{indexbirk, indexbaryossef}. Indeed, given a fixed ``side information'' reflecting the memory content of the nodes, the data delivery strategy for a particular permutation of the data rows induces an index coding problem. However, our coded shuffling framework is different from index coding in at least two significant ways.  First, the coded shuffling framework involves multiple iterations of data being stored across all the nodes.  Secondly, when the caches of the nodes are updated in coded shuffling, the system is unaware of the upcoming permutations.  Thus, the cache update rules need to be designed to target any possible unknown permutation of data in succeeding iterations of the algorithm.

We now review some recent works on coded shuffling, which have been published after our first presentation\,\cite{lee2015nips, lee2016isit}.
In \cite{attiaglobecom16}, the authors study the information-theoretic limits of the coded shuffling problem. 
More specifically, the authors completely characterize the fundamental limits for the case of $2$ workers and the case of $3$ workers.
In \cite{attiaallertonworst16}, the authors consider the worse-case formulation of the coded shuffling problem, and propose a two-stage shuffling algorithm.  
The authors of \cite{song2017pliable} propose a new coded shuffling scheme based on pliable index coding.
While most of the existing works focus on either coded computation or coded shuffling, one notable exception is\,\cite{songzeli2016shuffling}. 
In this work, the authors generalize the original coded MapReduce framework by introducing stragglers to the computation phases.
Observing that highly flexible codes are not favorable to coded shuffling while replication codes allow for efficient shuffling, the authors propose an efficient way of coding to mitigate straggler effects as well as reduce the shuffling overheads.

\section{Coded Computation}
\label{sec:computation}
In this section, we propose a novel paradigm to mitigate the straggler problem. The core idea is simple: \emph{we introduce redundancy into subtasks of a distributed algorithm such that the original task's result can be decoded from a subset of the subtask results, treating uncompleted subtasks as \textbf{erasures}}. 
For this specific purpose, we use \emph{erasure codes} to design \emph{coded} subtasks.

An erasure code is a method of introducing redundancy to information for robustness to noise\,\cite{cover2012elements}.
It encodes a message of $k$ symbols into a longer message of $n$ coded symbols such that the original $k$ message symbols can be recovered by decoding a subset of coded symbols\,\cite{costello2004error, cover2012elements}. 
We now show how erasure codes can be applied to distributed computation to mitigate the straggler problem.

\subsection{Coded Computation} \label{sec:coded_compute_def}
A coded distributed algorithm is specified by local functions, local data blocks, decodable sets of indices, and a decoding function:
The local functions and data blocks specify the way the original computational task and the input data are distributed across $n$ workers;
and the decodable sets of indices and the decoding function are such that the desired computation result can be correctly recovered using the decoding function as long as the local computation results from any of the decodable sets are collected.

The formal definition of coded distributed algorithms is as follows.
\begin{definition}[Coded computation]
Consider a computational task $f_\ba(\cdot)$.
A \emph{coded} distributed algorithm for computing $f_\ba(\cdot)$ is specified by 
\begin{itemize}
\item local functions $\langle f^i_{\ba_i}(\cdot) \rangle_{i=1}^{n}$ and local data blocks $\langle \ba_i\rangle_{i=1}^{n}$;
\item (minimal) decodable sets of indices $\mathcal{I} \subset \mathcal{P}([n])$ and a decoding function $\texttt{dec}(\cdot, \cdot)$,
\end{itemize}
where $[n] \defeq \{1,2,\ldots,n\}$, and $\mathcal{P}(\cdot)$ is the power set of a set.
The decodable sets of indices $\mathcal{I}$ is minimal: no element of $\mathcal{I}$ is a subset of other elements.
The decoding function takes a sequence of indices and a sequence of subtask results, and it must correctly output $f_\ba(\bsx)$ if any decodable set of indices and its corresponding 
results are given.
\end{definition}

A coded distributed algorithm can be run in a distributed computing cluster as follows. 
Assume that the \thth{i} (encoded) data block $\ba_i$ is stored at the \thth{i} worker for all $i$.  
Upon receiving the input argument $\bsx$, the master node multicasts $\bsx$ to all the workers, and then waits until it receives the responses from any of the decodable sets. 
Each worker node starts computing its local function when it receives its local input argument, and sends the task result to the master node. 
Once the master node receives the results from some decodable set, it decodes the received task results and obtains $f_\ba(\bsx)$.

%
%\begin{algorithm}[h]
%\begin{algorithmic}
%\On Receiving an input argument $\bsx$
%\State Multicast $\bsx$ to all the workers.
%\State $\mathbf{i} = \langle \rangle$
%\State $\bsy_{list} = \langle \rangle$
%\While {$\mathbf{i} \notin \mathcal{I}$}
%\On Receiving a message $\bsy$ from worker $j$
%\State $\mathbf{i} \gets \langle \mathbf{i}, j\rangle$
%\State $\bsy_{list} \gets \langle \bsy_{list}, \bsy \rangle$
%\EndWhile
%\State $\bsy \gets \texttt{dec}(\mathbf{i}, \bsy_{list})$
%\State Return $\bsy$
%\end{algorithmic}
%\caption{Coded computation: master node's protocol}
%\label{alg:protocol2}
%\end{algorithm}
%
%\begin{algorithm}[h]
%\begin{algorithmic}
%\On Receiving an input argument $\bsx$
%\State Compute $\bsy_i = f^i_{\ba_i}(\bsx)$
%\State Send $\bsy_i$ to the master node
%\end{algorithmic}
%\caption{Coded computation: worker node $i$'s protocol}
%\label{alg:protocol3}
%\end{algorithm}

The algorithm described in Sec.~\ref{sec:overview} is an example of coded distributed algorithms: it is a coded distributed algorithm for matrix multiplication that uses an $(n,n-1)$ MDS code. 
One can generalize the described algorithm using an $(n,k)$ MDS code as follows.
For any $1\leq k \leq n$, the data matrix $\ba$ is first divided into $k$ equal-sized submatrices\footnote{If the number of rows of $\ba$ is not a multiple of $k$, one can append zero rows to $\ba$ to make the number of rows a multiple of $k$.}. Then, by applying an $(n,k)$ MDS code to each element of the submatrices, $n$ encoded submatrices are obtained. 
We denote these $n$ encoded submatrices by $\ba'_1, \ba'_2, \ldots, \ba'_n$. 
Note that the $\ba'_i = \ba_i$ for $1\leq i\leq k$ if a systematic MDS code is used for the encoding procedure.
Upon receiving \emph{any} $k$ task results, the master node can use the decoding algorithm to decode $k$ task results. Then, one can find $\ba\bx$ simply by concatenating them.

\subsection{Runtime of Uncoded/Coded Distributed Algorithms}
In this section, we analyze the runtime of uncoded and coded distributed algorithms.
We first consider the overall runtime of an uncoded distributed algorithm, $T^\text{uncoded}_\text{overall}$.
Assume that the runtime of each task is identically distributed and independent of others.
We denote the runtime of the \thth{i} worker under a computation scheme, say s, by $T^\text{s}_i$.
Note that the distributions of $T_i$'s can differ across different computation schemes.
\begin{align} \label{eq:uncoded}
T^\text{uncoded}_\text{overall} =T^\text{uncoded}_{(n)} \defeq \max\{ T^\text{uncoded}_1, \ldots, T^\text{uncoded}_n \},
\end{align}
where $T_{(i)}$ is the \thth{i} smallest one in $\{T_i\}_{i=1}^{n}$.
From \eqref{eq:uncoded}, it is clear that a single straggler can slow down the overall algorithm.  
A \emph{coded} distributed algorithm is terminated whenever the master node receives results from any decodable set of workers. Thus, the overall runtime of a coded algorithm is \emph{not} determined by the slowest worker, but by the first time to collect results from some decodable set in $\mathcal{I}$, i.e.,
\begin{align} \label{eq:coded}
T^\text{coded}_\text{overall} =T^\text{coded}_{(\mathcal{I})} \defeq \min_{\mathbf{i} \in \mathcal{I}}{ \max_{j \in \mathbf{i}}{ T^\text{coded}_j }}
\end{align}
We remark that the runtime of uncoded distributed algorithms \eqref{eq:uncoded} is a special case of \eqref{eq:coded} with $\mathcal{I} = \{[n]\}$. In the following examples, we consider the runtime of the repetition-coded algorithms and the MDS-coded algorithms. 

\begin{example}[Repetition codes]
Consider an $\frac{n}{k}$-repetition-code where each local task is replicated $\frac{n}{k}$ times. 
We assume that each group of $\frac{n}{k}$ consecutive workers work on the replicas of one local task.
Thus, the decodable sets of indices $\mathcal{I}$ are all the minimal sets that have $k$ distinct task results, i.e., $\mathcal{I} = \{1,2, \ldots, \frac{n}{k}\} \times \{\frac{n}{k}+1, \frac{n}{k}+2, \ldots, \frac{n}{k}+k\} \times \ldots \times \{n-\frac{n}{k}+1, n-\frac{n}{k}+2, \ldots, n\}$, where $A \times B$ denotes the Cartesian product of matrix $A$ and $B$. Thus, 
\begin{align} \label{eq:rep_coded}
T^\text{Repetition-coded}_\text{overall} =\max_{i \in [k]} { \min_{j \in [\frac{n}{k}]} \{ T^\text{Repetition-coded}_{(i-1)\frac{n}{k}+j}\}  }.
\end{align}
\end{example}

\begin{example}[MDS codes]
If one uses an $(n,k)$ MDS code, the decodable sets of indices are the sets of \emph{any} $k$ indices, i.e., $\mathcal{I} = \{ \mathbf{i} | \mathbf{i} \subset [n],~|\mathbf{i}| = k \}$. Thus,
\begin{align} \label{eq:mds_coded}
T^\text{MDS-coded}_\text{overall} =T^\text{MDS-coded}_{(k)}
\end{align}
That is, the algorithm's runtime will be determined by the \thth{k} response, not by the \thth{n} response.
\end{example}

\subsection{Probabilistic Model of Runtime}
In this section, we analyze the runtime of uncoded/coded distributed algorithms assuming that task runtimes, including times to communicate inputs and outputs, are randomly distributed according to a certain distribution.
For analytical purposes, we make a few assumptions as follows. 
We first assume the existence of the \emph{mother runtime distribution} $F(t)$: we assume that running an algorithm using a \emph{single} machine takes a random amount of time $T_0$, that is a positive-valued, continuous random variable parallelized according to $F$, i.e. $\PP(T_0 \leq t) = F(t)$. 
We also assume that $T_0$ has a probability density function $f(t)$.
Then, when the algorithm is distributed into a certain number of subtasks, say $\ell$, the runtime distribution of each of the $\ell$ subtasks is assumed to be a scaled distribution of the mother distribution, i.e., $\PP(T_i \leq t) = F(\ell t)$ for $1 \leq i \leq \ell$. 
Note that we are implicitly assuming a \emph{symmetric} job allocation scheme, which is the optimal job allocation scheme if the underlying workers have the identical computing capabilities, i.e., homogeneous computing nodes are assumed.
Finally, the computing times of the $k$ tasks are assumed to be independent of one another.

\begin{remark}[Homogeneous Clusters and Heterogenous Clusters]
In this work, we assume homogeneous clusters: that is, all the workers have independent and identically distributed computing time statistics.
While our symmetric job allocation is optimal for homogeneous cases, it can be strictly suboptimal for heterogenous cases. 
While our work focuses on homogeneous clusters, we refer the interested reader to a recent work~\cite{codedcomputationhet} for a generalization of our problem setting to that of heterogeneous clusters, for which symmetric allocation strategies are no longer optimal.
\end{remark}

We first consider an uncoded distributed algorithm with $n$ (uncoded) subtasks. 
Due to the assumptions mentioned above, the runtime of each subtask is $F(nt)$. 
Thus, the runtime distribution of an uncoded distributed algorithm, denoted by $F_\text{overall}^{\text{uncoded}}(t)$, is simply $\left[F(nt) \right]^n$. 

When repetition codes or MDS codes are used, an algorithm is first divided into $k~(<n)$ systematic subtasks, and then $n-k$ coded tasks are designed to provide an appropriate level of redundancy. 
Thus, the runtime of each task is distributed according to $F(kt)$. 
Using \eqref{eq:rep_coded} and \eqref{eq:mds_coded}, 
one can easily find the runtime distribution of an $\frac{n}{k}$-repetition-coded distributed algorithm,  $F_\text{overall}^{\text{Repetition}}$, and the runtime distribution of an $(n,k)$-MDS-coded distributed algorithm, $F_\text{overall}^{\text{MDS-coded}}$. 
For an $\frac{n}{k}$-repetition-coded distributed algorithm, one can first find the distribution of $$\min_{j \in [\frac{n}{k}]} \{ T^\text{Repetition-coded}_{(i-1)\frac{n}{k}+j}\},$$ and then find the distribution of the maximum of $k$ such terms: 
\begin{align} \label{eq:F_rep}
F_\text{overall}^{\text{Repetition}}(t) = \left[1 - \left[1 - F(kt)\right]^{\frac{n}{k}}\right]^{k}.
\end{align}
The runtime distribution of an $(n,k)$-MDS-coded distributed algorithm is simply the \thth{k} order statistic: 
\begin{align} \label{eq:F_mds}
&F_\text{overall}^{\text{MDS-coded}}(t) \nonumber\\
&= \int_{\tau=0}^{t}{ nk f(k\tau) {n-1 \choose k-1} F(k\tau)^{k-1} \left[1-F(k\tau)\right]^{n-k} d\tau }.
\end{align}

\begin{remark}
For the same values of $n$ and $k$, the runtime distribution of a repetition-coded distributed algorithm strictly dominates that of an MDS-coded distributed algorithm. This can be shown by observing that the decodable sets of the MDS-coded algorithm contain those of the repetition-coded algorithm. 
\end{remark}

\begin{figure}[t]
\centering
\begin{subfigure}{0.45\textwidth}
\centering   
\includegraphics[width=\textwidth]{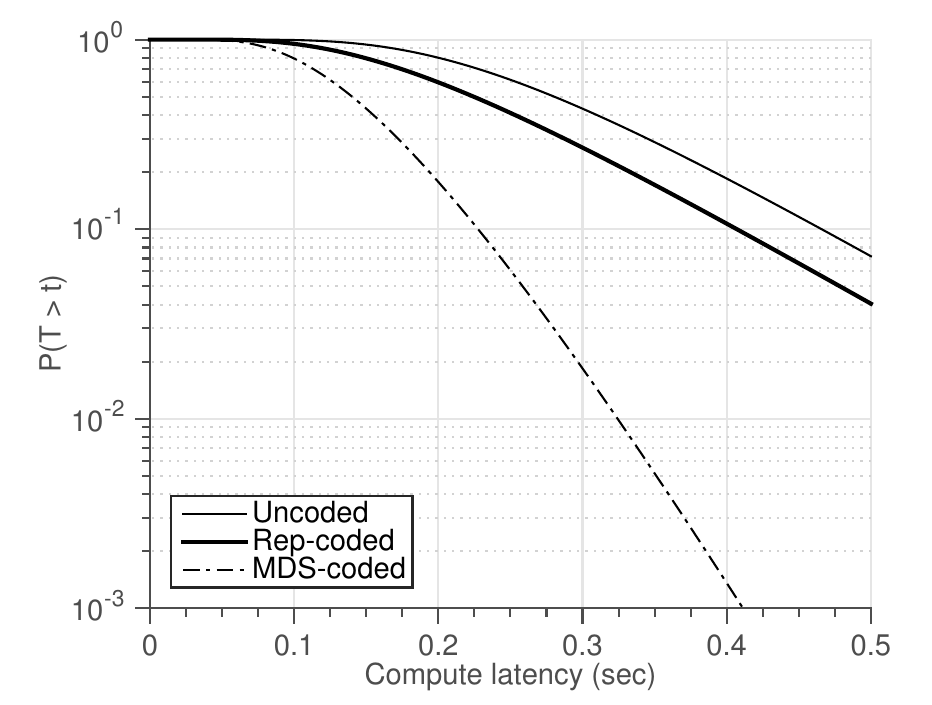}
\caption{\footnotesize{Shifted-exponential distribution}}
\label{fig:runtime_distribution_a}
\end{subfigure}
~~
\begin{subfigure}{0.45\textwidth}
\centering   
\includegraphics[width=\textwidth]{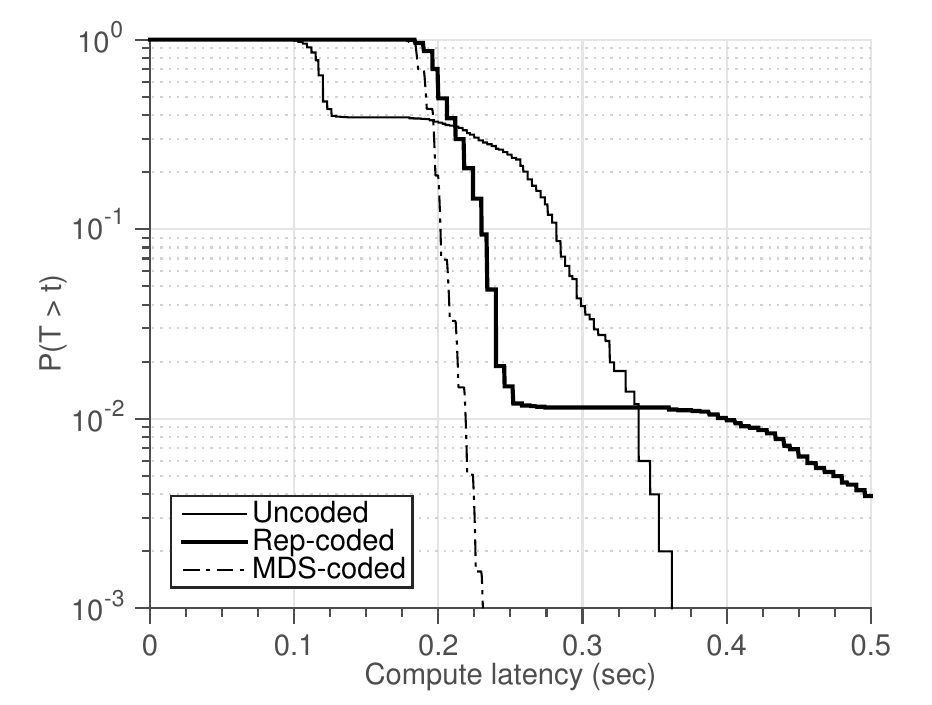}
\caption{\footnotesize{Empirical distribution}}
\label{fig:runtime_distribution_b}
\end{subfigure}
\caption{\footnotesize{\textbf{Runtime distributions of uncoded/coded distributed algorithms.} We plot the runtime distributions of uncoded/coded distributed algorithms. 
For the uncoded algorithms, we use $n=10$, and for the coded algorithms, we use $n=10$ and $k=5$.
In (a), we plot the runtime distribution when the runtime of tasks are distributed according to the shifted-exponential distribution. 
Indeed, the curves in (a) are analytically obtainable: See Sec.~\ref{sec:shifted_exp} for more details.
In (b), we use the empirical task runtime distribution measured on an Amazon EC2 cluster.
}}
\label{fig:runtime_distribution}
\end{figure}

In Fig.~\ref{fig:runtime_distribution}, we compare the runtime distributions of uncoded and coded distributed algorithms. We compare the runtime distributions of uncoded algorithm, repetition-coded algorithm, and MDS-coded algorithm with $n=10$ and $k=5$. In Fig.~\ref{fig:runtime_distribution_a}, we use a shifted-exponential distribution as the mother runtime distribution. That is, $F(t) = 1-e^{t-1}$ for $t\geq 1$. 
In Fig.~\ref{fig:runtime_distribution_b}, we use the empirical task runtime distribution that is measured on an Amazon EC2 cluster\footnote{The detailed description of the experiments is provided in Sec.~\ref{sec:simulation}.}. Observe that for both cases, the runtime distribution of the MDS-coded distribution has the lightest tail.

\subsection{Optimal Code Design for Coded Distributed Algorithms: The Shifted-exponential Case} \label{sec:shifted_exp}
When a coded distributed algorithm is used, the original task is divided into a fewer number of tasks compared to the case of uncoded algorithms. Thus, the runtime of each task of a coded algorithm, which is $F(kt)$, is stochastically larger than that of an uncoded algorithm, which is $F(nt)$.
If the value that we choose for $k$ is too small, then the runtime of each task becomes so large that the overall runtime of the distributed coded algorithm will eventually increase. 
If $k$ is too large, the level of redundancy may not be sufficient to prevent the algorithm from being delayed by the stragglers. 

Given the mother runtime distribution and the code parameters, one can compute the overall runtime distribution of the coded distributed algorithm using \eqref{eq:F_rep} and \eqref{eq:F_mds}. 
Then, one can optimize the design based on various target metrics, e.g., the expected overall runtime, the \thth{99} percentile runtime, etc. 

In this section, we show how one can design an optimal coded algorithm that minimizes \emph{the expected overall runtime} for a shifted-exponential mother distribution.
The shifted-exponential distribution  strikes a good balance between accuracy and analytical tractability. 
This model is motivated by the model proposed in \cite{liang13}: the authors used this distribution to model latency of file queries from cloud storage systems.
The shifted-exponential distribution is the sum of a constant and an exponential random variable, i.e., 
\begin{align}
\PP(T_0 \leq t) = 1 - e^{-\mu(t-1)},~~ \forall t \geq 1, \label{eq:shifted_exp}
\end{align}
where the exponential rate $\mu$ is called the \emph{straggling parameter}. 

With this shifted-exponential model, we first characterize a lower bound on the fundamental limit of the average runtime.
\begin{proposition}
The average runtime of any distributed algorithm, in a distributed computing cluster with $n$ workers, is lower bounded by $\frac{1}{n}$.
\end{proposition}
\begin{IEEEproof}
One can show that the average runtime of any distributed algorithm strictly decreases if the mother runtime distribution is replaced with a deterministic constant $1$. 
Thus, the optimal average runtime with this deterministic mother distribution serves as a strict lower bound on the optimal average runtime with the shifted-exponential mother distribution. 
The constant mother distribution implies that stragglers do not exist, and hence the uncoded distributed algorithm achieves the optimal runtime, which is $\frac{1}{n}$.
\end{IEEEproof}

We now analyze the average runtime of uncoded/coded distributed algorithms.
We assume that $n$ is large, and $k$ is linear in $n$. Accordingly, we approximate $H_n \defeq \sum_{i=1}^{n}{\frac{1}{i}} \simeq \log n$ and $H_{n-k} \simeq \log{(n-k)}$.
We first note that the expected value of the maximum of $n$ independent exponential random variables with rate $\mu$ is $\frac{H_n}{\mu}$. 
Thus, the average runtime of an uncoded distributed algorithm is
\begin{align}
\EE[T^\text{uncoded}_\text{overall}] = \frac{1}{n}\left( 1 + \frac{1}{\mu}\log{n}  \right) = \Theta\left(\frac{\log n}{n}\right).
\end{align}
For the average runtime of an $\frac{n}{k}$-Repetition-coded distributed algorithm, we first note that the minimum of $\frac{n}{k}$ independent exponential random variables with rate $\mu$ is distributed as an exponential random variable with rate $\frac{n}{k}\mu$. Thus, 
\begin{align}\label{eq:rep}
\EE[T^\text{Repetition-coded}_\text{overall}] = \frac{1}{k} \left( 1 + \frac{k}{n\mu}\log{k} \right) = \Theta\left(\frac{\log n}{n}\right).
\end{align}
Finally, we note that the expected value of the \thth{k} statistic of $n$ independent exponential random variables of rate $\mu$ is $\frac{H_{n} - H_{n-k}}{\mu}$. Therefore, 
\begin{align}
\EE[T^\text{MDS-coded}_\text{overall}] = \frac{1}{k} \left( 1 + \frac{1}{\mu}\log \left(\frac{n}{n-k} \right) \right) = \Theta\left(\frac{1}{n}\right). \label{eq:mds_coded_2}
\end{align}

Using these closed-form expressions of the average runtime, one can easily find the optimal value of $k$ that achieves the optimal average runtime. The following lemma characterizes the optimal repetition code for the repetition-coded algorithms and their runtime performances. 
\begin{lemma}[Optimal repetition-coded distributed algorithms] \label{thm:optimal_rep}
If $\mu \geq 1$, the average runtime of an $\frac{n}{k}$-Repetition-coded distributed algorithm, in a distributed computing cluster with $n$ workers, is minimized by setting $k = n$, i.e., \emph{not} replicating tasks. 
If $\mu = \frac{1}{v}$ for some integer $v > 1$, the average runtime is minimized by setting $k = \mu n$, and the corresponding minimum average runtime is $\frac{1}{n\mu}\left(1 + \log(n\mu)\right)$.
\end{lemma}
\begin{IEEEproof}
It is easy to see that \eqref{eq:rep} as a function of $k$ has a unique extreme point. 
By differentiating \eqref{eq:rep} with respect to $k$ and equating it to zero, we have $k = \mu n$. Thus, if $\mu \geq 1$, one should set $k = n$; if $\mu = \frac{1}{v} < 1$ for some integer $v$, one should set $k = \mu n$. 
\end{IEEEproof}
The above lemma reveals that the optimal repetition-coded distributed algorithm can achieve a lower average runtime than the uncoded distributed algorithm if $\mu < 1$; however, the optimal repetition-coded distributed algorithm still suffers from the factor of $\Theta(\log n)$, and cannot achieve the order-optimal performance.
The following lemma, on the other hand, shows that the optimal MDS-coded distributed algorithm can achieve the order-optimal average runtime performance. 
\begin{lemma}[Optimal MDS-coded distributed algorithms] \label{thm:optimal_mds}
The average runtime of an $(n,k)$-MDS-coded distributed algorithm, in a distributed computing cluster with $n$ workers, can be minimized by setting $k = k^\star$ where 
\begin{align}
k^{\star} = \left[1 + \frac{1}{W_{-1}(-e^{-\mu-1})} \right] n,
\end{align}
and $W_{-1}(\cdot)$ is the lower branch of Lambert W function\footnote{$W_{-1}(x)$, the lower branch of Lambert W function evaluated at $x$, is the unique solution of $te^{t} = x$ and $t \leq -1$.} Thus,
\begin{align}
T^\star \defeq \min_{k}~\EE[T^\text{MDS-coded}_\text{overall}] = \frac{-W_{-1}(-e^{-\mu-1})}{\mu n} \defeq \frac{\gamma^\star(\mu)}{n}.
\end{align}
\end{lemma}
\begin{IEEEproof}
It is easy to see that \eqref{eq:mds_coded_2} as a function of $k$ has a unique extreme point. 
By differentiating \eqref{eq:mds_coded_2} with respect to $k$ and equating it to zero, we have 
$\frac{1}{k^\star}\left(1 + \frac{1}{\mu}\log \left( \frac{n}{n-k^\star}  \right) \right) = \frac{1}{\mu} \frac{1}{n-k^\star}$.
By setting $k = \alpha^\star n$, we have $\frac{1}{\alpha^\star}\left(1 + \frac{1}{\mu}\log \left( \frac{1}{1-\alpha^\star}  \right) \right) = \frac{1}{\mu} \frac{1}{1-\alpha^\star}$, which implies $\mu +1 = \frac{1}{1-\alpha^\star} - \log \left( \frac{1}{1-\alpha^\star}  \right)$.
By defining $\beta = \frac{1}{1-\alpha^\star}$ and exponentiating both the sides, we have $e^{\mu+1} = \frac{e^{\beta}}{\beta}$.
Note that the solution of $\frac{e^x}{x} = t$, $t \geq e$ and $x \geq 1$ is $x=-W_{-1}(-\frac{1}{t})$. 
Thus, $\beta = -W_{-1}(-e^{-\mu -1})$.
By plugging the above equation into the definition of $\beta$, the claim is proved.
\end{IEEEproof}

We plot $nT^\star$ and $\frac{k^\star}{\mu}$ as functions of $\mu$ in Fig.~\ref{fig:alpha_and_gamma}.
\begin{figure}[t]
\centering
\begin{subfigure}{0.45\textwidth}
\centering
\includegraphics[width=\textwidth]{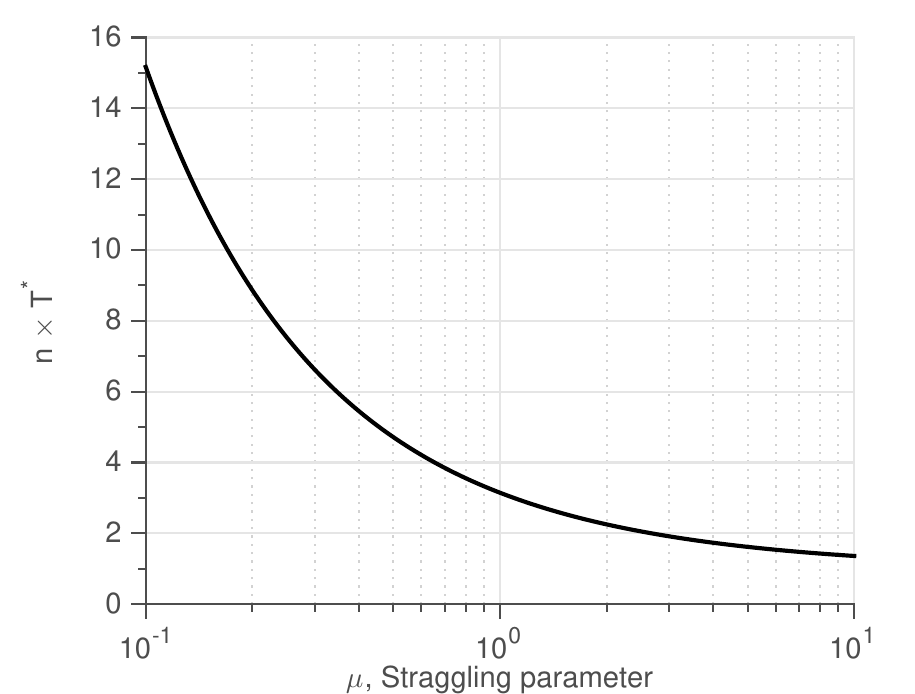}
\caption{\footnotesize{$nT^\star$ as a function of $\mu$.}}
\label{fig:mds_coeff}
\end{subfigure}
~~
\begin{subfigure}{0.45\textwidth}
\centering
\includegraphics[width=\textwidth]{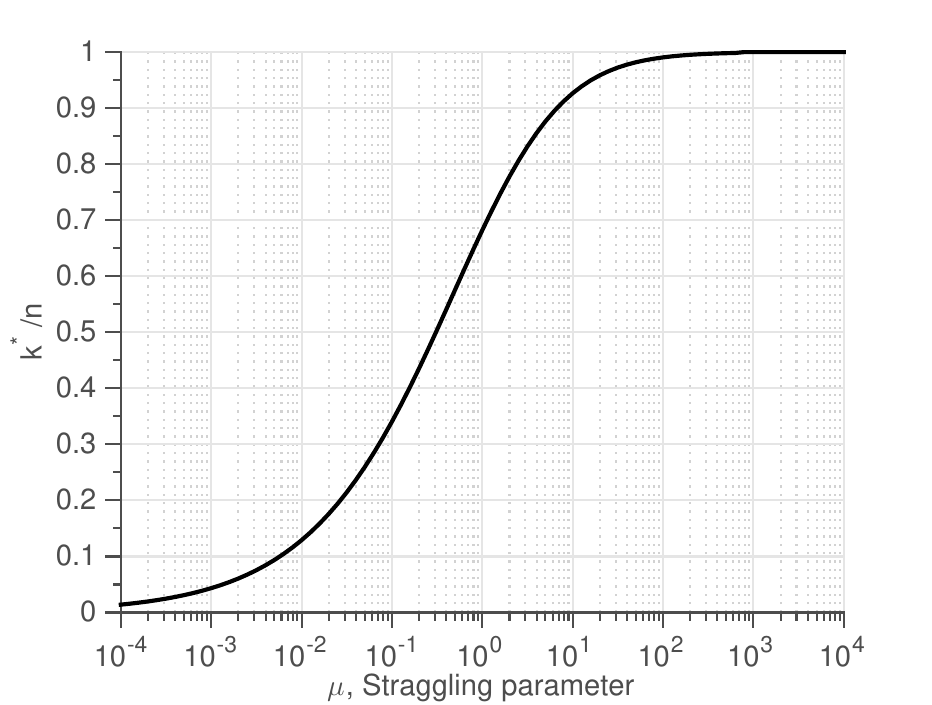}
\caption{\footnotesize{$\frac{k^\star}{n}$ as a function of $\mu$.}}
\label{fig:alpha_star}
\end{subfigure}
\caption{\footnotesize{\textbf{$nT^\star$ and $\frac{k^\star}{n}$ as functions of $\mu$.} As a function of the straggling parameter, we plot the normalized optimal computing time and the optimal value of $k$.
}}
\label{fig:alpha_and_gamma}
\end{figure} 
In addition to the order-optimality of MDS-coded distributed algorithms, the above lemma precisely characterizes the gap between the achievable runtime and the optimistic lower bound of $\frac{1}{n}$. 
For instance, when $\mu > 1$, the optimal average runtime is only $3.15$ away from the lower bound.

\begin{remark}[Storage overhead]
So far, we have considered only the runtime performance of distributed algorithms. 
Another important metric to be considered is the storage cost. 
When coded computation is being used, the storage overhead may increase.
For instance, the MDS-coded distributed algorithm for matrix multiplication, described in Sec.~\ref{sec:coded_compute_def}, requires $\frac{1}{k}$ of the whole data to be stored at each worker, while the uncoded distributed algorithm requires $\frac{1}{n}$. Thus, the storage overhead factor is $\frac{\frac{1}{k} - \frac{1}{n}}{\frac{1}{n}} = \frac{n}{k} - 1$. If one uses the runtime-optimal MDS-coded distributed algorithm for matrix multiplication, the storage overhead is $\frac{n}{k^\star}-1 = \frac{1}{\alpha^\star}-1$. 
\end{remark}

\subsection{Coded Gradient Descent: An MDS-coded Distributed Algorithm for Linear Regression} \label{sec:gradient_descent}
\begin{figure*}[t!]
\centering
\begin{subfigure}[t]{0.4\textwidth}
    \centering   
    \includegraphics[width=\textwidth]{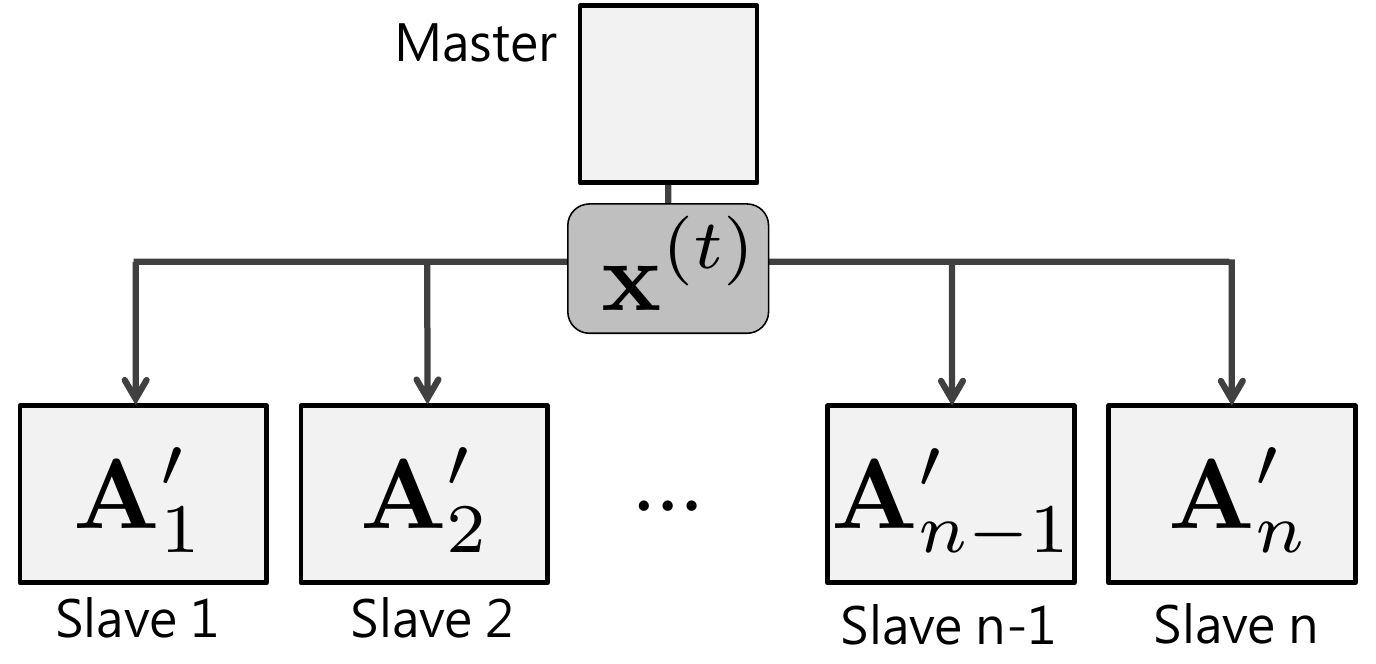}
  \caption{\footnotesize{In the beginning of the \thth{t} iteration, the master node multicasts $\bsx^{(t)}$ to the worker nodes. \label{fig:lr1}} }
\end{subfigure}
~~~~~~~~
\begin{subfigure}[t]{0.4\textwidth}
    \centering   
    \includegraphics[width=\textwidth]{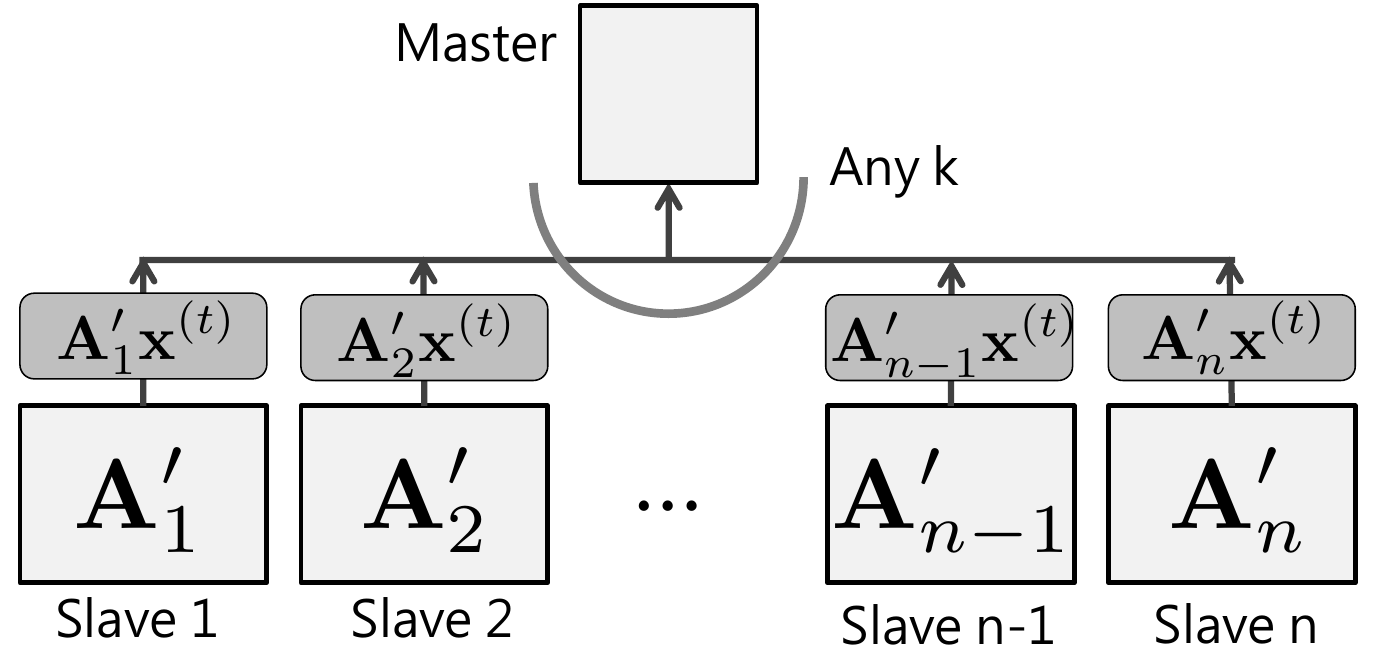}
  \caption{\footnotesize{The master node waits for the earliest responding $k$ worker nodes, and computes $\ba\bsx^{(t)}$. \label{fig:lr2}} }
\end{subfigure}
\\ \vspace{0.04in}
\begin{subfigure}[t]{0.4\textwidth}
    \centering   
\includegraphics[width=\textwidth]{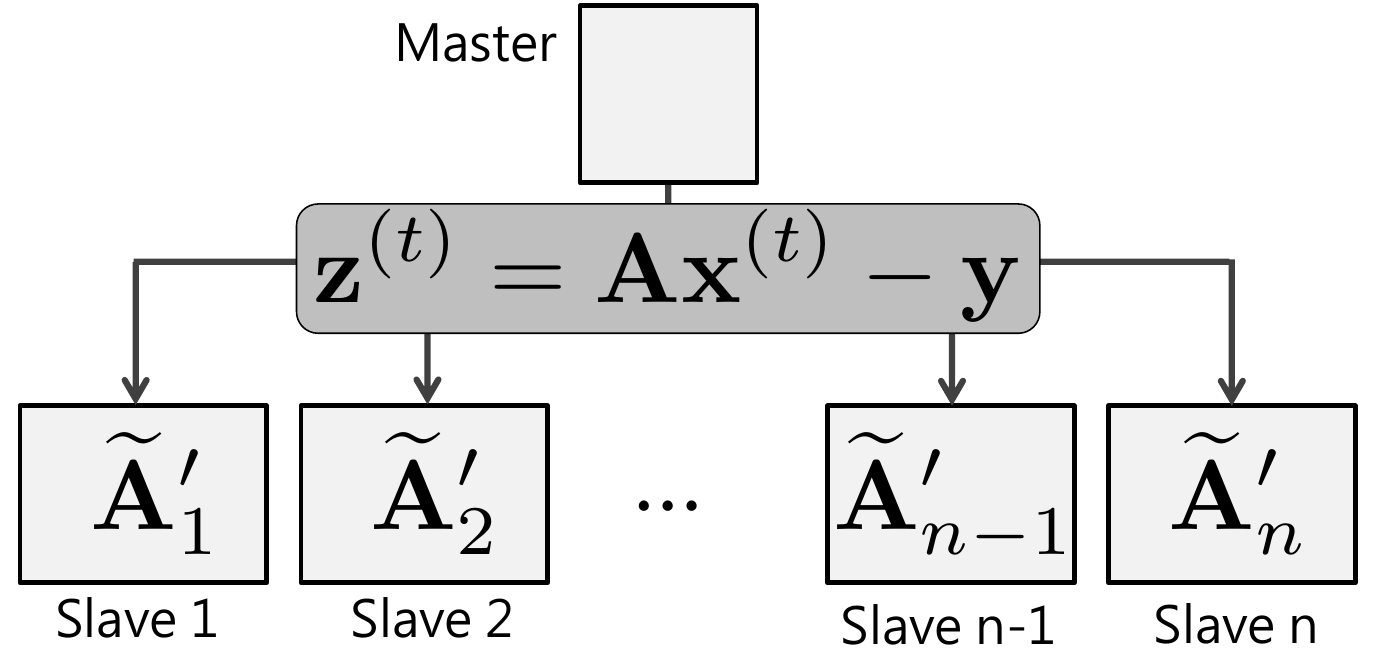}
  \caption{\footnotesize{The master node computes $\bsz^{(t)}$ and  multicasts it to the worker nodes. \label{fig:lr3}} }
\end{subfigure}
~~~~~~~~
\begin{subfigure}[t]{0.4\textwidth}
    \centering   
    \includegraphics[width=\textwidth]{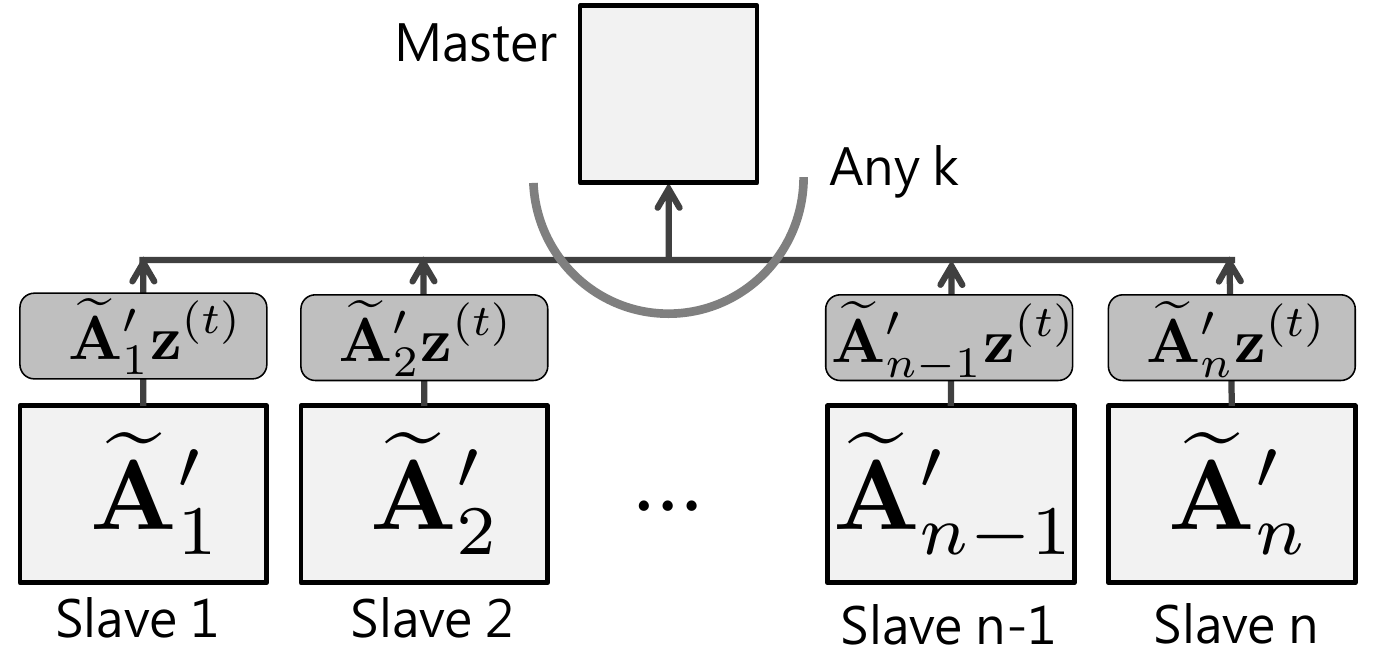}
  \caption{\footnotesize{The master node waits for the $k$ earliest responding worker nodes, and computes $\ba^T\bsz^{(t)}$ or $\nabla f(\bsx^{(t)})$. \label{fig:lr4}} }
\end{subfigure}
\caption{\footnotesize{\textbf{Illustration of a coded gradient descent approach for linear regression.} The coded gradient descent computes a gradient of the objective function using \emph{coded matrix multiplication} twice: in each iteration, it first computes $\ba \bsx^{(t)}$ as depicted in (a) and (b), and then computes $\ba^T (\ba\bsx^{(t)}-\bsy)$ as depicted in (c) and (d). 
}}
\label{fig:protocol_linear_regression}
\end{figure*}
In this section, as a concrete application of coded matrix multiplication, we propose the \emph{coded gradient descent} for solving large-scale linear regression problems. 

We first describe the (uncoded) gradient-based distributed algorithm. Consider the following linear regression,
\begin{align} \label{eq:linear_regression}
\min_\bsx f(\bsx) \defeq \min_{\bsx} \frac{1}{2}\|\ba\bsx - \bsy\|^2_2,
\end{align}
where $\bsy \in \mathbb{R}^{q}$ is the label vector, $\ba= [\bsa_1, \bsa_2, \ldots, \bsa_{q}]^T \in \mathbb{R}^{q \times r}$ is the data matrix, and $\bsx \in \mathbb{R}^{r}$ is the unknown weight vector to be found. 
We seek a distributed algorithm to solve this regression problem. 
Since $f(\bsx)$ is convex in $\bsx$, the gradient-based distributed algorithm works as follows. 
We first compute the objective function's gradient: $\nabla f (\bsx)= \ba^T(\ba\bsx-\bsy)$. Denoting by $\bsx^{(t)}$ the estimate of $\bsx$ after the \thth{t} iteration, we iteratively update $\bsx^{(t)}$ according to the following equation.
\begin{align}
\bsx^{(t+1)} = \bsx^{(t)} - \eta \nabla f(\bsx^{(t)}) = \bsx^{(t)} - \eta \ba^T (\ba\bsx^{(t)} - \bsy) \label{eq:exact_gradient}
\end{align}
The above algorithm is guaranteed to converge to the optimal solution if we use a small enough step size $\eta$  \cite{boyd2004convex}, and can be easily distributed. We describe one simple way of parallelizing the algorithm, which is implemented in many open-source machine learning libraries including Spark \texttt{mllib} \cite{spark_mllib}. 
As $\ba^T(\ba\bsx^{(t)}-\bsy)=\sum_{i=1}^{q}{ \bsa_i (\bsa_i^T \bsx^{(t)} - \bsy_i) }$, gradients can be computed in a distributed way by computing partial sums at different worker nodes and then adding all the partial sums at the master node. 
This distributed algorithm is an \emph{uncoded} distributed algorithm: in each round, the master node needs to wait for all the task results in order to compute the gradient.\footnote{Indeed, one may apply another coded computation scheme called Gradient Coding\,\cite{tandon2016gradient}, which was proposed after our conference publications. By applying Gradient Coding to this algorithm, one can achieve straggler tolerance but at the cost of significant computation and storage overheads. More precisely, it incurs $\Theta(n)$ larger computation and storage overheads in order to protect the algorithm from $\Theta(n)$ stragglers. Later in this section, we will show that our coded computation scheme, which is tailor-designed for linear regression, incurs $\Theta(1)$ overheads to protect the algorithm from $\Theta(n)$ stragglers.}
Thus, the runtime of each update iteration is determined by the slowest response among all the worker nodes.

We now propose the \emph{coded gradient descent}, a coded distributed algorithm for linear regression problems. 
Note that in each iteration, the following two matrix-vector multiplications are computed.
\begin{align}
\ba\bsx^{(t)}, ~~\ba^T (\ba\bsx^{(t)} - \bsy) \defeq \ba^T \bsz^{(t)}
\end{align}
In Sec.~\ref{sec:coded_compute_def}, we proposed the MDS-coded distributed algorithm for matrix multiplication.
Here, we apply the algorithm twice to compute these two multiplications in each iteration.
More specifically, for the first matrix multiplication, we choose $1 \leq k_1 < n$ and use an $(n,k_1)$-MDS-coded distributed algorithm for matrix multiplication to encode the data matrix $\ba$.
Similarly for the second matrix multiplication, we choose $1 \leq k_2 < n$ and use a $(n,k_2)$-MDS-coded distributed algorithm to encode the transpose of the data matrix.
Denoting the \thth{i} row-split (column-split) of $\ba$ as $\ba_i$ ($\widetilde{\ba}_i$), 
the \thth{i} worker stores both $\ba_i$ and $\widetilde{\ba}_i$.
In the beginning of each iteration, the master node multicasts $\bsx^{(t)}$ to the worker nodes, each of which computes the local matrix multiplication for $\ba\bsx^{(t)}$ and sends the result to the master node. 
Upon receiving \emph{any} $k_1$ task results, the master node can start decoding the result and obtain $\bsz^{(t)} = \ba \bsx^{(t)}$. 
The master node now multicasts $\bsz^{(t)}$ to the workers, and the workers compute local matrix multiplication for $\ba^T \bsz^{(t)}$. 
Finally, the master node can decode $\ba^T \bsz^{(t)}$ as soon as it receives any $k_2$ task results, and can proceed to the next iteration. 
Fig.~\ref{fig:protocol_linear_regression} illustrates the protocol with $k_1 = k_2 = n-1$. 

\begin{remark}[Storage overhead of the coded gradient descent]
The coded gradient descent requires each node to store a $(\frac{1}{k_1} + \frac{1}{k_2} - \frac{1}{k_1 k_2})$-fraction of the data matrix. As the minimum storage overhead per node is a $\frac{1}{n}$-fraction of the data matrix, the relative storage overhead of the coded gradient descent algorithm is at least about factor of $2$, if $k_1 \simeq n$ and $k_2 \simeq n$. 
\end{remark}

\subsection{Experimental Results} \label{sec:simulation}
In order to see the efficacy of coded computation, we implement the proposed algorithms and test them on an Amazon EC2 cluster. We first obtain the empirical distribution of task runtime in order to observe how frequently stragglers appear in our testbed by measuring round-trip times between the master node and each of $10$ worker instances on an Amazon EC2 cluster. 
Each worker computes a matrix-vector multiplication and passes the computation result to the master node, and the master node measures round trip times that include both computation time and communication time. 
Each worker repeats this procedure $500$ times, and we obtain the empirical distribution of round trip times across all the worker nodes. 
\begin{figure}[h]
\centering   
\includegraphics[width=0.45\textwidth]{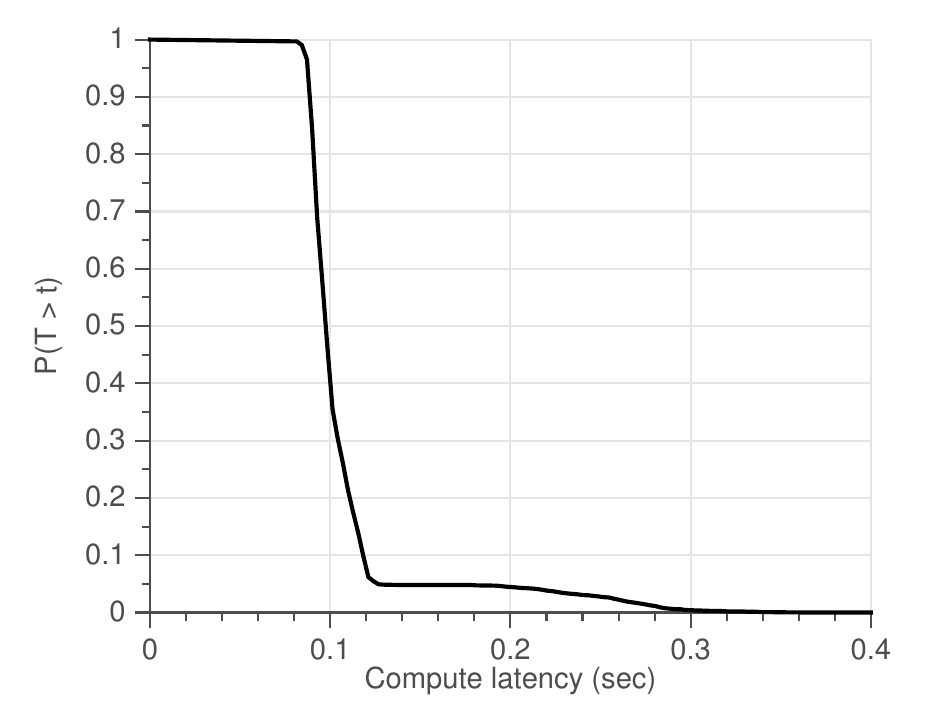}
\caption{\footnotesize{\textbf{Empirical CCDF of the measured round trip times.} We measure round trip times between the master node and each of $10$ worker nodes on an Amazon EC2 cluster. A round trip time consists of transmission time of the input vector from the master to a worker, computation time, and transmission time of the output vector from a worker to the master.}}
\label{fig:measurements}
\end{figure}
%\caption{\footnotesize{Empirical CCDF of the measured round trip times.}}

%\centering
%\begin{subfigure}{0.45\textwidth}
%\centering   
%\includegraphics[width=\textwidth]{figs/hist.eps}
%\caption{\footnotesize{Histogram of the measured round trip times}}
%\label{fig:hist}
%\end{subfigure}
%~~
%\begin{subfigure}{0.45\textwidth}

In Fig.~\ref{fig:measurements}, we plot the histogram and complementary CDF (CCDF) of measured computing times; the average round trip time is $0.11$ second, and the \thth{95} percentile latency is $0.20$ second, i.e., roughly five out of hundred tasks are going to be roughly two times slower than the average tasks.
Assuming the probability of a worker being a straggler is $5\%$, if one runs an uncoded distributed algorithm with $10$ workers, the probability of not seeing such a straggler is only about $60\%$, so the algorithm is slowed down by a factor of more than $2$ with probability $40\%$.
Thus, this observation strongly emphasizes the necessity of an efficient straggler mitigation algorithm.  
In Fig.~\ref{fig:runtime_distribution_a}, we plot the runtime distributions of uncoded/coded distributed algorithms using this empirical distribution as the mother runtime distribution. 
When an uncoded distributed algorithm is used, the overall runtime distribution entails a heavy tail, while the runtime distribution of the MDS-coded algorithm has almost no tail.

\begin{figure*}[t!]
\centering
\begin{subfigure}{0.255\textwidth}
\includegraphics[width=\textwidth]{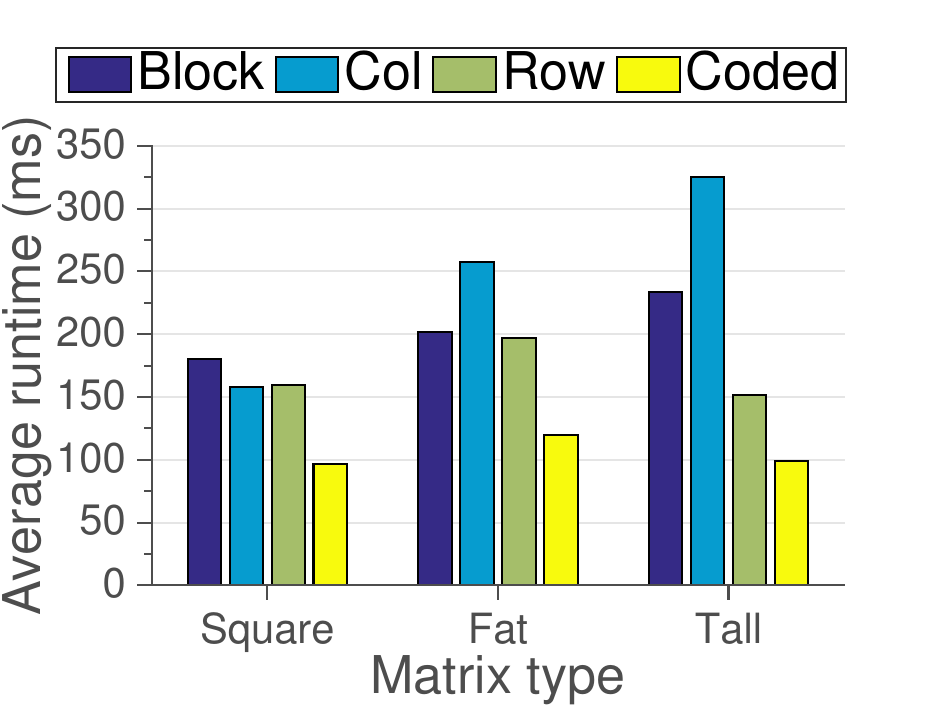}
\caption{\footnotesize{\texttt{m1-small}, average runtime}}
\label{fig:mm_1}
\end{subfigure}
\!\!\!\!\!\!\!
\begin{subfigure}{0.255\textwidth}
\includegraphics[width=\textwidth]{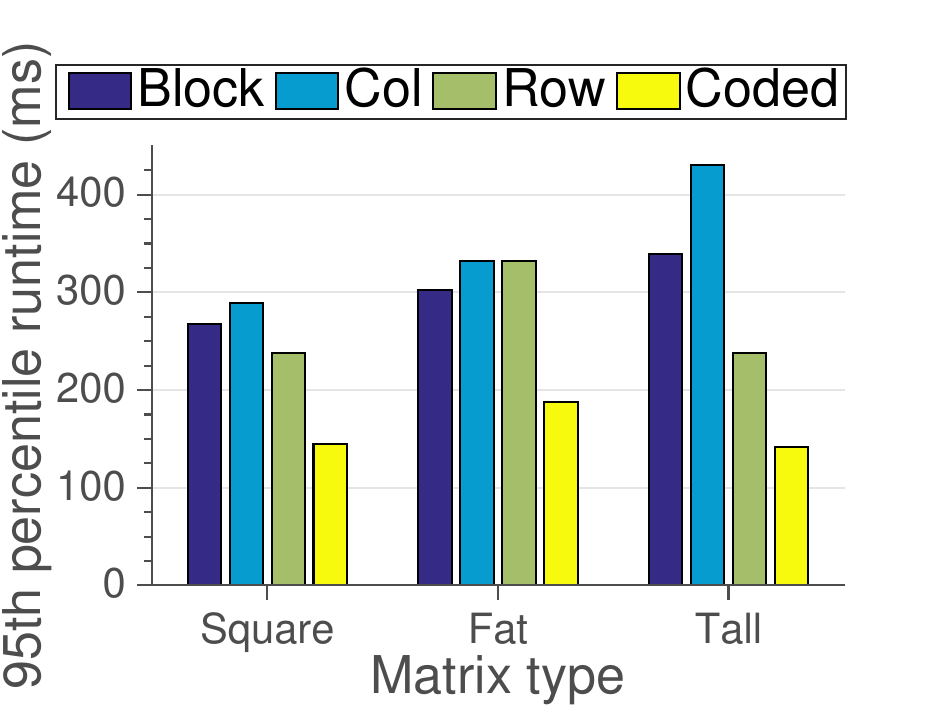}
\caption{\footnotesize{\texttt{m1-small}, tail runtime}}
\label{fig:mm_2}
\end{subfigure}
\!\!\!\!\!\!\!
\begin{subfigure}{0.255\textwidth}
\includegraphics[width=\textwidth]{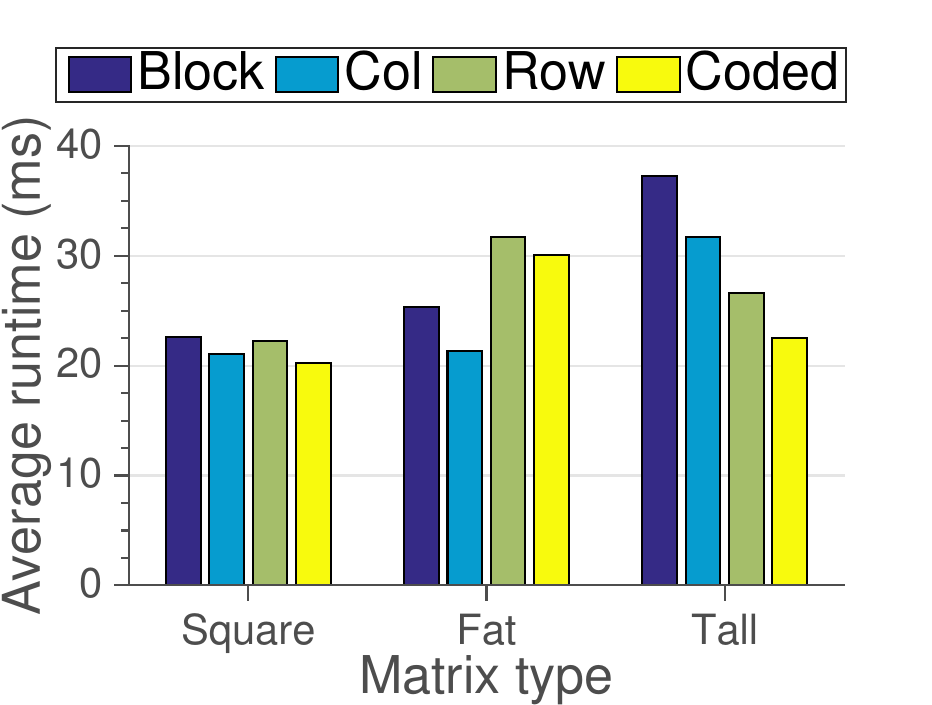}
\caption{\footnotesize{\texttt{c1-med}, average runtime}}
\label{fig:mm_3}
\end{subfigure}
\!\!\!\!\!\!\!
\begin{subfigure}{0.255\textwidth}
\includegraphics[width=\textwidth]{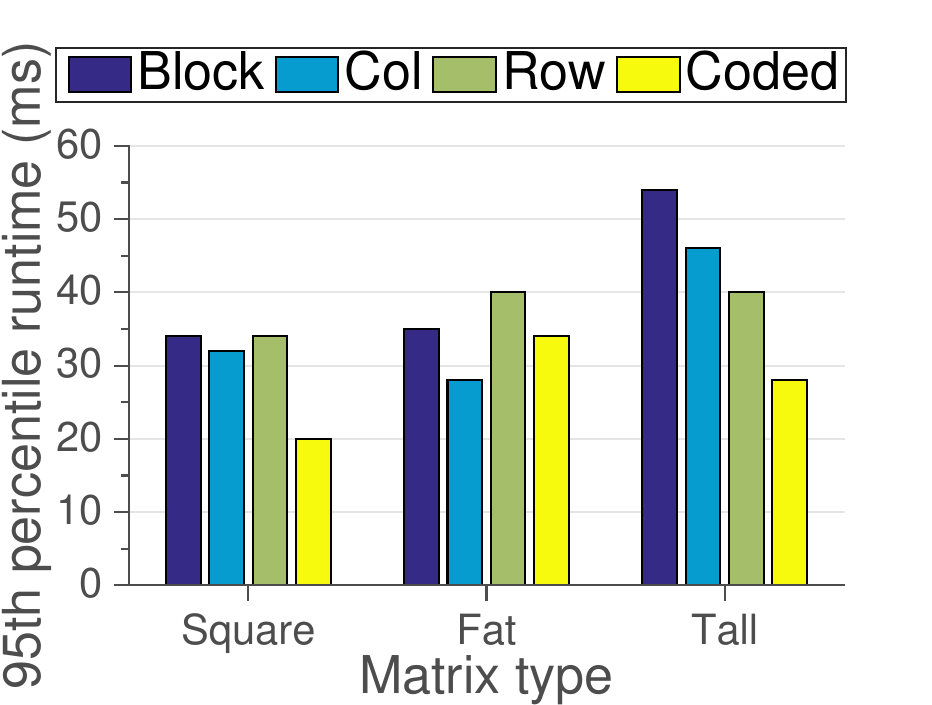}
\caption{\footnotesize{\texttt{c1-med}, tail runtime}}
\label{fig:mm_4}
\end{subfigure}
\caption{\footnotesize{\textbf{Comparison of parallel matrix multiplication algorithms.} We compare various parallel matrix multiplication algorithms: block, column-partition, row-partition, and coded (row-partition) matrix multiplication. We implement the four algorithms using OpenMPI and test them on Amazon EC2 cluster of $25$ instances. We measure the average and the \thth{95} percentile runtime of the algorithms.
Plotted in (a) and (b) are the results with \texttt{m1-small} instances, and in (c) and (d) are the results with \texttt{c1-medium} instances.}}
\label{fig:mm}
\end{figure*} 
We then implement the coded matrix multiplication in C++ using OpenMPI\cite{openmpi}​ and benchmark on a cluster of $26$ EC2 instances ($25$ workers and a master)\footnote{For the benchmark, we manage the cluster using the StarCluster toolkit \cite{starcluster}. Input data is generated using a Python script, and the input matrix is row-partitioned for each of the workers (with the required encoding as described in the previous sections) in a preprocessing step.
The procedure begins by having all of the worker nodes read in their respective row-partitioned matrices. Then, the master node reads the input vector and distributes it to all worker nodes in the cluster through an asynchronous send (\texttt{MPI\_Isend}). Upon receiving the input vector, each worker node begins matrix multiplication through a BLAS \cite{blas} routine call and once completed sends the result back to the master using \texttt{MPI\_Send}.  The master node waits for a sufficient number of results to be received by continuously polling (\texttt{MPI\_Test}) to see if any results are obtained. The procedure ends when the master node decodes the overall result after receiving enough partial results.}.
Also, three uncoded matrix multiplication algorithms -- block, column-partition, and row-partition -- are implemented and benchmarked. 

We randomly draw a square matrix of size $5750\times 5750$, a fat matrix of size $5750\times 11500$, and a tall matrix of size $11500\times 5750$, and multiply them with a column vector. 
For the coded matrix multiplication, we choose an $(25,23)$ MDS code so that the runtime of the algorithm is not affected by any $2$ stragglers. 
Fig.~\ref{fig:mm} shows that the coded matrix multiplication outperforms all the other parallel matrix multiplication algorithms in most cases. 
On a cluster of \texttt{m1-small}, the most unreliable instances, the coded matrix multiplication achieves about $40\%$ average runtime reduction and about $60\%$ tail reduction compared to the best of the $3$ uncoded matrix multiplication algorithmss.
On a cluster of \texttt{c1-medium} instances, the coded algorithm achieves the best performance in most of the tested cases: the average runtime is reduced by at most $39.5\%$, and the \thth{95} percentile runtime is reduced by at most $58.3\%$.
Among the tested cases, we observe one case in which both the uncoded row-partition and the coded row-partition algorithms are outperformed by the uncoded column-partition algorithm. 
This is the case of a fat matrix multiplication with \texttt{c1-medium} instances. 
Note that when a row-partition algorithm is used, the size of messages from the master node to the workers is $n$ times larger compared with the case of column-partition algorithms. 
Thus, when the variability of computational times becomes low compared with that of communication time, the larger communication overhead of row-partition algorithms seems to arise, nullifying the benefits of coding.

\begin{figure}[h]
\centering
\begin{subfigure}{0.24\textwidth}
\includegraphics[width=\textwidth]{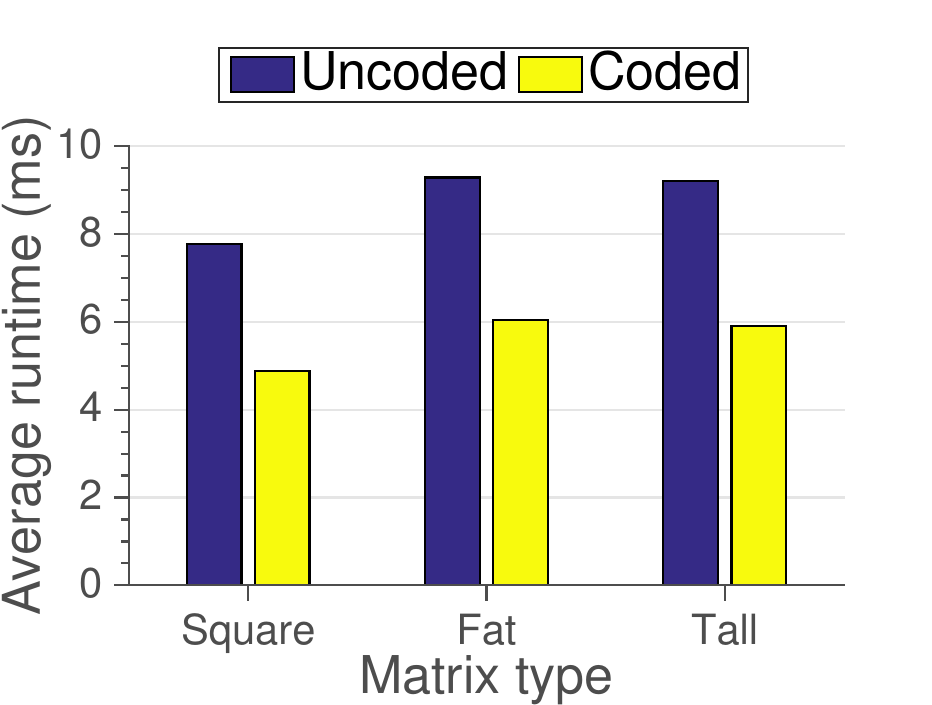}
\caption{\footnotesize{Average runtime}}
\end{subfigure}
\!\!\!\!\!\!\!
\begin{subfigure}{0.24\textwidth}
\includegraphics[width=\textwidth]{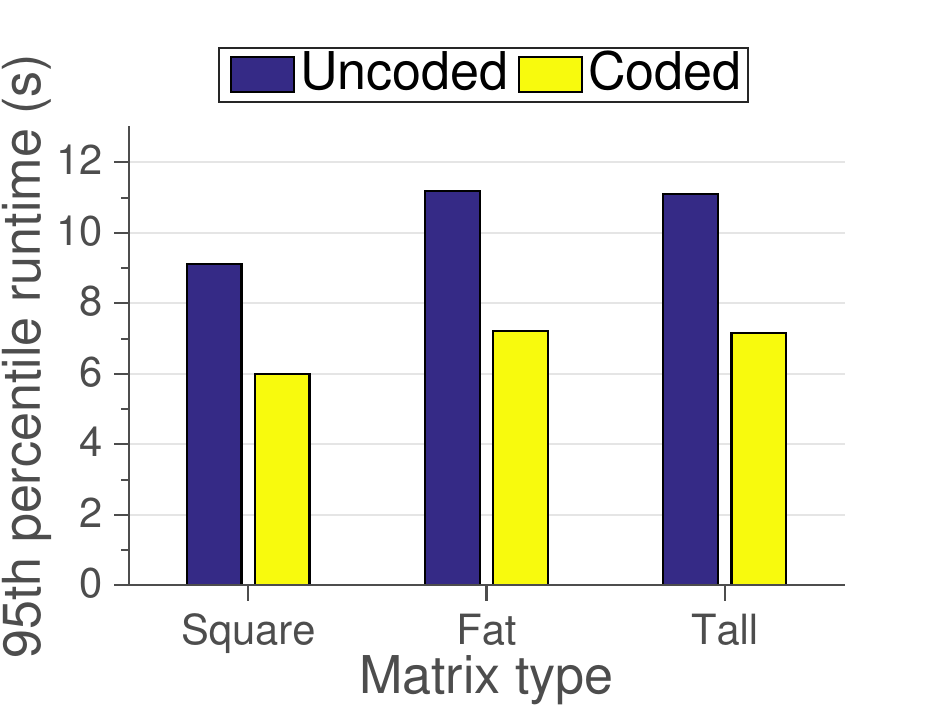}
\caption{\footnotesize{Tail runtime}}
\end{subfigure}
\caption{\footnotesize{\textbf{Comparison of parallel gradient algorithms.} We compare parallel gradient algorithms for linear regression problems. We implement both the uncoded gradient descent algorithm and the coded gradient descent algorithm using Open MPI, and test them on an Amazon EC2 cluster of $10$ worker instances. Plotted are the average and the \thth{95} percentile runtimes of the algorithms.\label{fig:lr}}}
\end{figure}

We also evaluate the performance of the coded gradient descent algorithm for linear regression.
The coded linear regression procedure is also implemented in C++ using OpenMPI, and benchmarked on a cluster of $11$ EC2 machines ($10$ workers and a master).
Similar to the previous benchmarks, we randomly draw a square matrix of size $2000\times 2000$, a fat matrix of size $400\times 10000$, and a tall matrix of size $10000\times 400$, and use them as a data matrix. 
We use a $(10,8)$-MDS code for the coded linear regression so that each multiplication of the gradient descent algorithm is not slowed down by up to $2$ stragglers.
Fig.~\ref{fig:lr} shows that the gradient algorithm with the \emph{coded matrix multiplication} significantly outperforms the one with the uncoded matrix multiplication; the average runtime is reduced by $31.3\%$ to $35.7\%$, and the tail runtime is reduced by $27.9\%$ to $35.6\%$.

\section{Coded Shuffling}
\label{sec:shuffling}
We shift our focus from solving the straggler problem to solving the communication bottleneck problem. 
In this section, we explain the problem of data-shuffling, propose the \emph{Coded Shuffling} algorithm, and analyze its performance.
%In Section~\ref{sec:sim}, we provide extensive simulation and experiment results that corroborate our theoretical findings. 

\subsection{Setup and Notations}
We consider a master-worker distributed setup, where the master node has access to the entire data-set.
Before every {\it iteration} of the distributed algorithm, the master node randomly partition the entire data set into $n$ subsets, say $\ba_1, \ba_2, \ldots, \ba_n$. 
The goal of the shuffling phase is to distribute each of these partitioned data sets to the corresponding worker so that each worker can perform its distributed task with its own exclusive data set after the shuffling phase.

We let $\ba(\mathcal{J}) \in \mathbb{R}^{| \mathcal{J}| \times r}, ~ \mathcal{J} \subset [q]$ be the concatenation of $|\mathcal J |$ rows of matrix $\ba$ with indices in $\mathcal J$. Assume that each worker node has a cache of size $s$ data rows (or $s \times r$ real numbers). In order to be able to fully store the data matrix across the worker nodes, we impose the inequality condition $q/n \leq s$. Further, clearly if $s > q$, the data matrix can be fully stored at each worker node, eliminating the need for any shuffling. Thus, without loss of generality we assume that $s \leq q$. As explained earlier working on the same data points at each worker node in all the iterations of the iterative optimization algorithm leads to slow convergence. Thus, to enhance the statistical efficiency of the algorithm, the data matrix is shuffled after each iteration. More precisely, at each iteration $t$, the set of data rows $[q]$ is partitioned uniformly at random into $n$ subsets $S^t_i, ~ 1 \leq i \leq n$ so that $\cup_{i=1}^{n} S^t_i = [q]$ and $S^t_i \cap S^t_j = \emptyset$ when $i \neq j$; thus, each worker node computes a fresh local function of the data.  Clearly, the data set that worker $i$ works on has cardinality $q/n$, i.e., $|S^t_i| = q / n$. 
Note that the sampling we consider here is {\em without replacement}, and hence these data sets are non-overlapping.

\subsection{Shuffling Schemes} 
We now present our coded shuffling algorithm, consisting of a transmission strategy for the master node, and caching and decoding strategies for the worker nodes.
Let $C_i^t$ be the cache content of node $i$ (set of row indices stored in cache $i$) at the end of iteration $t$. We design a transmission algorithm (by the master node) and a cache update algorithm to ensure that (i)
$S_i^t\subset C_i^t$;
and (ii) $C_i^t \setminus S_i^t$ is distributed uniformly at random without replacement in the set $[q] \setminus S_i^t$. The first condition ensures that at each iteration, the workers have access to the data set that they are supposed to work on. The second condition provides the opportunity of effective coded transmissions for shuffling in the next iteration as will be explained later.

\subsubsection{Cache Update Rule}
We consider the following cache update rule: the new cache will contain the subset of the data points used in the current iteration (this is needed for the local computations), plus a random subset of the previous cached contents.
More specifically, $q/n$ rows of the new cache are  precisely the rows in $S_i^{t+1}$, and $s - q/n$ rows of the cache are sampled points from the set $C_i^t \setminus S_i^{t+1} $, uniformly at random without replacement. 
Since the permutation $\pi^t$ is picked uniformly at random, the marginal distribution of the cache contents at iteration $t+1$ given $S^{t+1}_i, ~1 \leq i \leq n$ is described as follows: $S^{t+1}_i \subset C^{t+1}_i$ and $C^{t+1}_i \setminus S^{t+1}_i$ is distributed uniformly at random in $[q] \setminus S^{t+1}_i$ without replacement.

\subsubsection{Encoding and Transmission Schemes}
We now formally describe two transmission schemes of the master node: (1) uncoded transmission and (2) coded transmission. 
In the following descriptions, we drop the iteration index $t$ (and $t+1$) for the ease of notation. 

The uncoded transmission first finds how many data rows in $S_i$ are already cached in $C_i$, i.e. $|C_i \cap S_i|$. Since, the new permutation (partitioning) is picked uniformly at random, $s/q$ fraction of the data row indices in $S_i$ are cached in $C_i$, so as $q$ gets large, we have $|C_i \cap S_i| + o(q)= \frac{q}{n}(1 - s/q)$. Thus, without coding, the master node needs to transmit $\frac{q}{n}(1 - s/q)$ data points to each of the $n$ worker nodes. The total communication rate (in data points transmitted per iteration) of the uncoded scheme is then 
\begin{align}\label{eq:r1}
R_{u} = n \times \frac{q}{n}(1 - s/q) = q (1 - s/q).
\end{align}

We now describe the coded transmission scheme.
Define the set of ``exclusive'' cache content as $\widetilde{C}_{\mathcal{I}} = \left(\cap_{i \in \mathcal{I}} C_i \right)\cap \left(\cap_{i' \in [n] \setminus \mathcal{I}} C^\complement_{i'}\right)$ that denotes the set of rows that are stored at the caches of $\mathcal{I}$, and are {\em not} stored at the caches of $[n] \setminus \mathcal{I}$. 
For each subset $\mathcal{I}$ with $|\mathcal{I}| \geq 2$, the master node will multicast $\sum_{i \in \mathcal{I}} \ba(S_i \cap \widetilde{C}_{\mathcal{I} \setminus \{ i\}})$ to the worker nodes. 
Note that in general, the matrices $\ba$'s differ in their sizes, so one has to zero-pad the shorter matrices and sum the zero-padded matrices. 
%Note that the size of such metadata is negligible compared to the size of coded messages.
Algorithm~\ref{alg:enc} provides the pseudocode of the coded encoding and transmission scheme.\footnote{Note that for each encoded data row, the master node also needs to transmit tiny metadata describing which data rows are included in the summation. We omit this detail in the description of the algorithm.}
\begin{algorithm}[h]
  \begin{algorithmic}
  \Procedure{Encoding}{$[C_i]_{i=1}^{n}$}
  \For{each $\mathcal{I} \in [n]^n$, $|\mathcal{I}|>2$}
  \State $\widetilde{C}_{\mathcal{I}} = \left(\cap_{i \in \mathcal{I}} C_i \right)\cap \left(\cap_{i' \in [n] \setminus \mathcal{I}} C^\complement_{i'}\right)$
  \State $\ell \gets \max_{i=1}^{|\mathcal{I}|} {|S_i \cap \widetilde{C}_{\mathcal{I} \setminus \{ i\}}|}$
  \For{each $i \in \mathcal{I}$}
  \State $\mathbf{B}_i[1:|S_i \cap \widetilde{C}_{\mathcal{I} \setminus \{ i\}}|,:] \gets \ba(S_i \cap \widetilde{C}_{\mathcal{I} \setminus \{ i\}})$
  \State $\mathbf{B}_i[|S_i \cap \widetilde{C}_{\mathcal{I} \setminus \{ i\}}|+1:\ell,:] \gets \mathbf{0}$
  \EndFor   
  \State broadcast $\sum_{i \in \mathcal{I}} \mathbf{B}_i$
  \EndFor
  \EndProcedure
  \end{algorithmic}
  \caption{Coded Encoding and Transmission Scheme}
\label{alg:enc}
\end{algorithm}

\subsubsection{Decoding Algorithm}
The decoding algorithm for the uncoded transmission scheme is straightforward: each worker simply takes the additional data rows that are required for the new iteration, and ignores the other data rows. 
We now describe the decoding algorithm for the coded transmission scheme.
Each worker, say worker $i$, decodes each encoded data row as follows. 
Consider an encoded data row for some $\mathcal{I}$ that contains $i$.
(All other data rows are discarded.)
Such an encoded data row must be the sum of some data row in $S_i$ and $|\mathcal{I}|-1$ data rows in $\widetilde{C}_{\mathcal{I}\setminus\{i\}}$, which are available in worker $i$ by the definition of $\widetilde{C}$. 
Hence, the worker can always subtract the data rows corresponding to $\widetilde{C}_{\mathcal{I}\setminus\{i\}}$ and decode the data row in $S_i$.

\subsection{Example}
The following example illustrates the coded shuffling scheme.

\begin{example}
Let $n = 3$. Recall that worker node $i$ needs to obtain $\ba(S_i \cap C^\complement_i)$ for the next iteration of the algorithm. Consider $i = 1$. The data rows in $S_1 \cap C^\complement_1$ are stored either exclusively in $C_2$ or $C_3$ (i.e. $\widetilde{C}_2$ or $\widetilde{C}_3$), or stored in both $C_2$ and $C_3$ (i.e. $\widetilde{C}_{2,3}$). The transmitted message consists of 4 parts:
\begin{itemize}
\item{(Part $1$)} ${M}_{\{1,2\}} = \ba(S_1 \cap \widetilde{C}_2) + \ba(S_2 \cap \widetilde{C}_1)$,
\item{(Part $2$)} ${ M}_{\{1,3\}} = \ba(S_1 \cap \widetilde{C}_3) + \ba(S_3 \cap \widetilde{C}_1)$,
\item{(Part $3$)} ${ M}_{\{2,3\}} = \ba(S_2 \cap \widetilde{C}_3) + \ba(S_3 \cap \widetilde{C}_2)$, and
\item{(Part $4$)} ${ M}_{\{1,2,3\}} = \ba(S_1 \cap \widetilde{C}_{2,3}) + \ba(S_2 \cap \widetilde{C}_{1,3}) + \ba(S_3 \cap \widetilde{C}_{1,2})$.
\end{itemize}
We show that worker node 1 can recover the data rows that it does not store or $\ba(S_1 \cap C^\complement_1)$. First, observe that node $1$ stores $S_2 \cap \widetilde{C}_1$. Thus, it can recover $\ba(S_1 \cap \widetilde{C}_2)$ using part 1 of the message since $\ba(S_1 \cap \widetilde{C}_2) = { M}_1 - \ba(S_2 \cap \widetilde{C}_1)$. Similarly, node $1$ recovers $\ba(S_1 \cap \widetilde{C}_3) = {M}_2 - \ba(S_3 \cap \widetilde{C}_1)$. Finally, from part 4 of the message, node $1$ recovers $\ba(S_1 \cap \widetilde{C}_{2,3}) = { M}_4 -  \ba(S_2 \cap \widetilde{C}_{1,3}) - \ba(S_3 \cap \widetilde{C}_{1,2})$.

\end{example}

\subsection{Main Results} \label{sec:shuffle_main_results}
We now present the main result of this section, which characterizes the communication rate of the coded scheme. Let $p = \frac{s -q/n}{q - q/n}$.   
\begin{theorem}[Coded Shuffling Rate]\label{thm1}
Coded shuffling achieves communication rate 
\begin{align}\label{eq:r2}
R_{c} =  \frac{q}{(np)^2}\left ( (1-p)^{n+1} +  (n-1)p(1-p)  - (1-p)^2 \right )
\end{align}
(in number of data rows transmitted per iteration from the master node), which is significantly smaller than $R_u$ in \eqref{eq:r1}. 
\end{theorem}
The reduction in communication rate is illustrated in Fig.~\ref{fig:rate} for $n = 50$ and $q = 1000$ as a function of $s/q$, where $1/n \leq s/q \leq 1$. 
\begin{figure}
\centering
\includegraphics[width=0.5\textwidth]{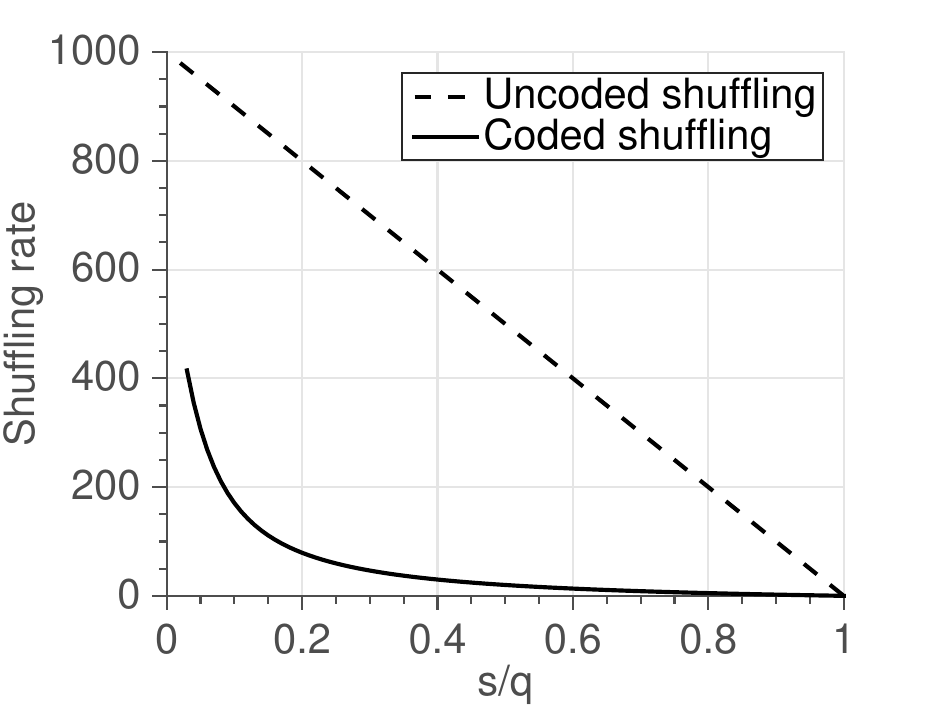}
\caption{\footnotesize{\textbf{The achievable rates of coded and uncoded shuffling schemes.} This figure shows the achievable rates of coded and uncoded schemes versus the cache size for parallel stochastic gradient descent algorithm.\label{fig:rate}}}
\end{figure}
For instance, when $s/q = 0.1$, the communication overhead for data-shuffling is reduced by more than $81\%$.
Thus, at a very low storage overhead for caching, the algorithm can be significantly accelerated.

Before we present the proof of the theorem, we briefly compare our main result with similar results shown in\,\cite{MN1, songzeliisit16}. 
Our coded shuffling algorithm is related to the coded caching problem\,\cite{MN1}, since one can design the right cache update rule to reduce the communication rate for an unknown demand or permutation of the data rows. A key difference though is that the coded shuffling algorithm is run over many iterations of the machine learning algorithm. Thus, the right cache update rule is required to guarantee the opportunity of coded transmission at every iteration. Furthermore, the coded shuffling problem has some connections to coded MapReduce\,\cite{songzeliisit16} as both algorithms mitigate the communication bottlenecks in distributed computation and machine learning. However, coded shuffling enables coded transmission of raw data by leveraging the extra memory space available at each node, while coded MapReduce enables coded transmission of processed data in the shuffling phase of the MapReduce algorithm by cleverly introducing redundancy in the computation of the mappers.

We now prove Theorem~\ref{thm1}.
\begin{proof}
To find the transmission rate of the coded scheme we first need to find the cardinality of sets $S^{t+1}_i \cap \widetilde{C}^t_{\mathcal{I}}$ for $\mathcal{I} \subset [n]$ and $i \notin \mathcal{I}$. To this end, we first find the probability that a random data row, $\bsr$, belongs to $\widetilde{C}^t_{\mathcal{I}}$. Denote this probability by $\PP(\bsr \in \widetilde{C}^t_{\mathcal{I}})$. Recall that the cache content distribution at iteration $t$: $q/n$ rows of cache $j$ are stored with $S^t_j$ and the other $s - q/n$ rows are stored uniformly at random. Thus, we can compute $\PP(\bsr \in \widetilde{C}^t_{\mathcal{I}})$ as follows. 
\begin{align}
&\PP(\bsr \in \widetilde{C}^t_{\mathcal{I}})\nonumber\\ \label{1}
&= 
\sum_{i=1}^n \PP(\bsr \in \widetilde{C}^t_{\mathcal{I}} | \bsr \in S^t_i) \PP(\bsr \in S^t_i) \\\label{2}
&= \sum_{i=1}^n \frac 1n \PP(\bsr \in \widetilde{C}^t_{\mathcal{I}} | \bsr \in S^t_i) \\ \label{3}
&= \sum_{i \in \mathcal{I}} \frac 1n \PP(\bsr \in \widetilde{C}^t_{\mathcal{I}}| \bsr \in S^t_i) \\ \label{4}
&= \sum_{i \in \mathcal{I}} \frac 1n \left ( \frac{s - q/n}{q - q/n}\right )^{|\mathcal{I}| - 1}\left( 1 - \frac{s - q/n}{q - q/n}\right)^{n-|\mathcal{I}|}\\  \label{5}
&= \frac{|\mathcal{I}|}{n}p^{|\mathcal{I}| - 1}(1-p)^{n - |\mathcal{I}|}.
\end{align}    
\eqref{1} is by the law of total probability. \eqref{2} is by the fact that $\bsr$ is chosen randomly. To see \eqref{3}, note that $\PP(\bsr \in \widetilde{C}^t_{\mathcal{I}}|\bsr \in S_i^t,i \notin \mathcal{I}) = 0$. Thus, the summation can be written only on the indices of $\mathcal{I}$. We now explain \eqref{4}. Given that $\bsr$ belongs to $S_i^t$, and $i \in \mathcal{I}$, then $\bsr \in C_i$ with probability 1. The other $|\mathcal{I}| - 1$ caches with indices in $\mathcal{I} \setminus \{ i \}$ contain $\bsr$ with probability $\frac{s - q/n}{q - q/n}$ independently. Further, the caches with indices in $[n] \setminus \mathcal{I}$ do not contain $\bsr$ with probability $1 - \frac{s - q/n}{q - q/n}$. By defining $p \defeq \frac{s - q/n}{q - q/n}$, we have \eqref{5}.

We now find the cardinality of $S^{t+1}_i \cap \widetilde{C}^t_{\mathcal{I}}$ for $\mathcal{I} \subset [n]$ and $i \notin \mathcal{I}$. Note that $|S^{t+1}_i| = q/n$. Thus, as $q$ gets large (and $n$ remains sub-linear in $q$), by the law of large numbers, 
\begin{align}
|S^{t+1}_i \cap \widetilde{C}^t_{\mathcal{I}}| = \frac qn \times \frac{|\mathcal{I}|}{n}p^{|\mathcal{I}| - 1}(1-p)^{n - |\mathcal{I}|} + o(q).
\end{align}
Recall that for each subset $\mathcal{I}$ with $|\mathcal{I}| \geq 2$, the master node will send $\sum_{i \in \mathcal{I}} \ba(S_i \cap \widetilde{C}_{\mathcal{I} \setminus \{i\}})$ . Thus, the total rate of coded transmission is 
\begin{equation}\label{rate1}
R_c = \sum_{i =2}^n {n \choose i} \frac{q}{n} \frac{i-1}{n}p^{i - 2}(1-p)^{n - (i-1)}.
\end{equation}
To complete the proof, we simplify the above expression. Let $x = \frac{p}{1-p}$. 
Taking derivative with respect to $x$ from both sides of the equality $\sum_{i=1}^n {n \choose i} x^{i-1} = \frac{1}{x} \left [ (1+x)^n - 1 \right ]$, we have 
\begin{equation}\label{eq:comb}
\sum_{i=2}^n {n \choose i} (i-1) x^{i-2} = \frac{1 + (1+x)^{n-1}(nx - x -1)}{x^2}.
\end{equation}
Using \eqref{eq:comb} in \eqref{rate1} completes the proof.   
\end{proof}

\begin{corollary}\label{cor:2}
Consider the case that the cache sizes are just enough to store the data required for processing; that is $s = q/n$. Then, $R_c = \frac 12 R_u$. Thus, one gets a factor 2 reduction gain in communication rate by exploiting coded caching.
\end{corollary}

Note that when $s = q/n$, $p = 0$. Finding the limit $\lim_{p \to 0} R_c$ in \eqref{eq:r2}, after some manipulations, one calculates 
\begin{align}
R_c = q\left(1 - \frac sq\right)\frac{1}{1 + ns/q} = R_u /2,
\end{align}
which shows Corollary~\ref{cor:2}.

\begin{corollary} \label{corollary:shuffling_scale}
Consider the regime of interest where $n$, $s$, and $q$ get large, and $s/ q \to c > 0 $ and $n / q \to 0$. Then, 
\begin{align}\label{eq:r3}
R_{c} \to q\left(1 - \frac sq\right)\frac{1}{ns/q} = \frac{R_u}{ns/q}
\end{align}
Thus, using coding, the communication rate is reduced by $\Theta(n)$.
\end{corollary}

\begin{remark}[The advantage of using multicasting over unicasting]\label{remark_tree}
\begin{figure}[h]
\centering
\includegraphics[width=.5\textwidth]{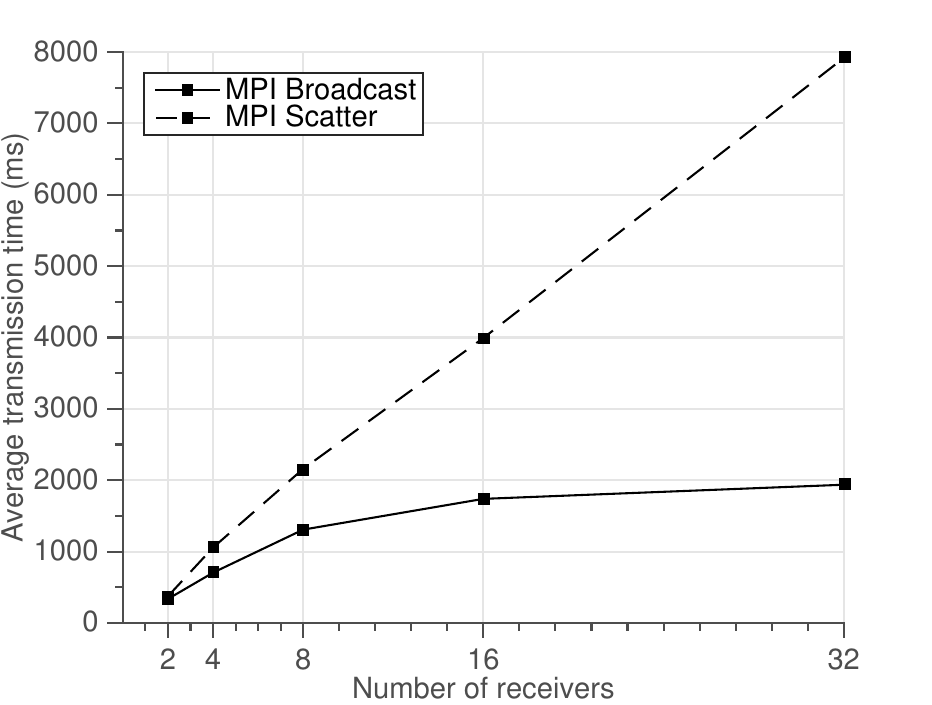}
\caption{\footnotesize{\textbf{Gains of multicasting over unicasting in distributed systems.} 
We measure the time taken for a data block of size of $4.15$ MB to be transmitted to a targeted number of workers on an Amazon EC2 cluster, and compare the average transmission time taken with Message Passing Interface (MPI) scatter (unicast) and that with MPI broadcast. 
Observe that the average transmission time increases linearly as the number of receivers increases, but with MPI broadcast, the average transmission time increases logarithmically.}
}
\label{fig:broadcast}
\end{figure}
It is reasonable to assume that $\gamma(n) \simeq n$ for wireless architecture that is of great interest with the emergence of wireless data centers, e.g. \cite{wirelessdc,zhu2014}, and mobile computing platforms \cite{mobilecomputing}.
However, still in many applications, the network topology is based on point-to-point communication, and the multicasting opportunity is not fully available, i.e., $\gamma(n) < n$. 
For these general cases, we have to renormalize the communication cost of coded shuffling since we have assumed that $\gamma(n) = n$ in our results.
For instance, in the regime considered in Corollary \ref{corollary:shuffling_scale}, the renormalized communication cost of coded shuffling $R^\gamma_c$ given $\gamma(n)$ is 
\begin{align}
R^\gamma_c = \frac{n}{\gamma(n)} R_c \rightarrow \frac{R_u}{\gamma(n)s/q}.
\end{align}
Thus, the communication cost of coded shuffling is smaller than uncoded shuffling if $\gamma(n) > q/s$. Note that $s/q$ is the fraction of the data matrix that can be stored at each worker's cache. Thus, in the regime of interest where $s/q$ is a constant independent of $n$, and $\gamma(n)$ scales with $n$, the reduction gain of coded shuffling in communication cost is still unbounded and increasing in $n$.  

We emphasize that even in point-to-point communication networks, \emph{multicasting the same message to multiple nodes is significantly faster than unicasting different message (of the same size) to multiple nodes}, i.e., $\gamma(n) \gg 1$, justifying the advantage of using coded shuffling. 
For instance, the MPI broadcast API (\texttt{MPI\_Bcast}) utilizes a tree multicast algorithm, which achieves $\gamma(n) = \Theta\left(\frac{n}{\log {n}}\right)$. 
Shown in Fig.~\ref{fig:broadcast} is the time taken for a data block to be transmitted to an increasing number of workers on an Amazon EC2 cluster, which consists of a point-to-point communication network.
We compare the average transmission time taken with MPI scatter (unicast) and that with MPI broadcast. 
Observe that the average transmission time increases linearly as the number of receivers increases, but with MPI broadcast, the average transmission time increases logarithmically.
\end{remark}

\section{Conclusion}\label{sec:conclusion}  
In this paper, we have explored the power of coding in order to make distributed algorithms robust to a variety of sources of ``system noise'' such as stragglers and communication bottlenecks. 
We propose a novel \emph{Coded Computation} framework that can significantly speed up existing distributed algorithms, by introducing redundancy through codes into the computation.
Further, we propose \emph{Coded Shuffling} that can significantly reduce the heavy price of data-shuffling, which is required for achieving high statistical efficiency in distributed machine learning algorithms. 
Our preliminary experimental results validate the power of our proposed schemes in effectively curtailing the negative effects of system bottlenecks, and attaining significant speedups of up to $40\%$, compared to the current state-of-the-art methods.

There exists a whole host of theoretical and practical open problems related to the results of this paper.
For coded computation, instead of the MDS codes, one could achieve different tradeoffs by employing another class of codes. 
Then, although matrix multiplication is one of the most basic computational blocks in many analytics, it would be interesting to leverage coding for a broader class of distributed algorithms.

For coded shuffling, convergence analysis of distributed machine learning algorithms under shuffling is not well understood.  
As we observed in the experiments, shuffling significantly reduces the number of iterations required to achieve a target reliability, but missing is a rigorous analysis that compares the convergence performances of algorithms with shuffling or without shuffling. Further, the trade-offs between bandwidth, storage, and the statistical efficiency of the distributed algorithms are not well understood. 
Moreover, it is not clear how far our achievable scheme, which achieves a bandwidth reduction gain of $\Theta(\frac{1}{n})$, is from the fundamental limit of communication rate for coded shuffling. 
Therefore, finding an information-theoretic lower bound on the rate of coded shuffling is another interesting open problem.

\bibliographystyle{IEEEtran}
\bibliography{ref}

\end{document}